%% file: main.tex
\documentclass[11pt,letterpaper]{article}

\def\showauthornotes{0}
\def\showtableofcontents{0}
\def\showkeys{0}
\def\showdraftbox{0}
\def\showcolorlinks{1}
\def\usemicrotype{1}
\def\showfixme{0}

\def\writemode{0}

\input{preamble/macros}
\input{preamble/macros-custom}
\input{preamble/ryan-macros}

\renewcommand{\S}{\mathbb{S}}
\newcommand{\rnote}[1]{}
\newcommand{\CONSMALL}{\scalebox{.75}[1.0]{\textnormal{SMALL}}}
\newcommand{\cons}{\mathrm{cons}}
\newcommand{\cl}{\mathrm{cl}}

\title{SOS lower bounds with hard constraints: \\ think global, act local}

\author{Pravesh K. Kothari\thanks{Princeton University and Institute for Advanced Study. \texttt{kothari@cs.princeton.edu}.} \and Ryan O'Donnell\thanks{Computer Science Department, Carnegie Mellon University.   \texttt{odonnell@cs.cmu.edu}. Some work performed at the Bo\u{g}azi\c{c}i University Computer Engineering Department, supported by Marie Curie International Incoming Fellowship project number 626373.  Also supported by NSF grants CCF-1618679, CCF-1717606. This material is based upon work supported by the
National Science Foundation under grant numbers listed above. Any opinions,
findings and conclusions or recommendations expressed in this material are
those of the author and do not necessarily reflect the views of the National
Science Foundation (NSF). } \and Tselil Schramm\thanks{UC Berkeley.	\texttt{tschramm@cs.berkeley.edu}. This work was partly supported by an NSF Graduate Research Fellowship (1106400), and also by a Simons Institute Fellowship.}}

\begin{document}

\maketitle
 \draftbox
\thispagestyle{empty}

\input{content/abstract}

\clearpage

\ifnum\showtableofcontents=1
{
\tableofcontents
\thispagestyle{empty}
 }
\fi

\clearpage

\setcounter{page}{1}

\input{content/intro}
\input{content/prelims}
\input{content/expansion}
\input{content/exact-objective}
\input{content/splitting}

\input{content/min-bi-feige}
\input{content/bisection}
\input{content/conclusions}

\addreferencesection
\bibliographystyle{amsalpha}{}
\bibliography{bib/mathreview,bib/dblp,bib/scholar,bib/custom,bib/odonnell-bib}

\appendix

\input{content/kmow-appendix}

\input{content/expansion-appendix}

\input{content/imbalanced}

\end{document}

%% file: preamble/macros.tex
\usepackage{etex}

\usepackage[l2tabu, orthodox]{nag}

\usepackage{xspace,enumerate}

\usepackage[dvipsnames]{xcolor}

\usepackage[T1]{fontenc}
\usepackage[full]{textcomp}

\usepackage[american]{babel}

\usepackage{mathtools}

\usepackage{amsthm}

\newtheorem{theorem}{Theorem}[section]
\newtheorem*{theorem*}{Theorem}

\newtheorem{proposition}[theorem]{Proposition}
\newtheorem*{proposition*}{Proposition}
\newtheorem{lemma}[theorem]{Lemma}
\newtheorem*{lemma*}{Lemma}
\newtheorem{corollary}[theorem]{Corollary}
\newtheorem*{conjecture*}{Conjecture}
\newtheorem{fact}[theorem]{Fact}
\newtheorem*{fact*}{Fact}

\newtheorem*{hypothesis*}{Hypothesis}

\theoremstyle{definition}
\newtheorem{definition}[theorem]{Definition}

\newtheorem{remark}[theorem]{Remark}
\newtheorem*{remark*}{Remark}

\theoremstyle{remark}
\newtheorem{claim}[theorem]{Claim}
\newtheorem*{claim*}{Claim}

\newtheorem*{observation*}{Observation}

\ifnum\writemode=1
\usepackage[
letterpaper,
top=1.2in,
bottom=1.2in,
left=1in,
right=1in]{geometry}

\fi

\ifnum\writemode=0
\usepackage[
letterpaper,
top=0.7in,
bottom=0.9in,
left=1in,
right=1in]{geometry}
\fi

\usepackage{newpxtext} 
\usepackage{textcomp} 
\usepackage[varg,bigdelims]{newpxmath}
\usepackage[scr=rsfso]{mathalfa}
\usepackage{bm} 
\let\mathbb\varmathbb

\ifnum\showkeys=1
\usepackage[color]{showkeys}
\fi

\ifnum\showcolorlinks=1
\usepackage[
pagebackref,
colorlinks=true,
urlcolor=blue,
linkcolor=blue,
citecolor=OliveGreen,
]{hyperref}
\fi

\ifnum\showcolorlinks=0
\usepackage[
pagebackref,
colorlinks=false,
pdfborder={0 0 0}
]{hyperref}
\fi

 \usepackage{prettyref}

 \newcommand{\savehyperref}[2]{\texorpdfstring{\hyperref[#1]{#2}}{#2}}

 \newrefformat{eq}{\savehyperref{#1}{\textup{(\ref*{#1})}}}
 \newrefformat{eqn}{\savehyperref{#1}{\textup{(\ref*{#1})}}}
 \newrefformat{lem}{\savehyperref{#1}{Lemma~\ref*{#1}}}
 \newrefformat{def}{\savehyperref{#1}{Definition~\ref*{#1}}}
 \newrefformat{thm}{\savehyperref{#1}{Theorem~\ref*{#1}}}
 \newrefformat{cor}{\savehyperref{#1}{Corollary~\ref*{#1}}}
 \newrefformat{cha}{\savehyperref{#1}{Chapter~\ref*{#1}}}
 \newrefformat{sec}{\savehyperref{#1}{Section~\ref*{#1}}}
 \newrefformat{app}{\savehyperref{#1}{Appendix~\ref*{#1}}}
 \newrefformat{tab}{\savehyperref{#1}{Table~\ref*{#1}}}
 \newrefformat{fig}{\savehyperref{#1}{Figure~\ref*{#1}}}
 \newrefformat{hyp}{\savehyperref{#1}{Hypothesis~\ref*{#1}}}
 \newrefformat{alg}{\savehyperref{#1}{Algorithm~\ref*{#1}}}
 \newrefformat{rem}{\savehyperref{#1}{Remark~\ref*{#1}}}
 \newrefformat{item}{\savehyperref{#1}{Item~\ref*{#1}}}
 \newrefformat{step}{\savehyperref{#1}{step~\ref*{#1}}}
 \newrefformat{conj}{\savehyperref{#1}{Conjecture~\ref*{#1}}}
 \newrefformat{fact}{\savehyperref{#1}{Fact~\ref*{#1}}}
 \newrefformat{prop}{\savehyperref{#1}{Proposition~\ref*{#1}}}
 \newrefformat{prob}{\savehyperref{#1}{Problem~\ref*{#1}}}
 \newrefformat{claim}{\savehyperref{#1}{Claim~\ref*{#1}}}
 \newrefformat{relax}{\savehyperref{#1}{Relaxation~\ref*{#1}}}
 \newrefformat{red}{\savehyperref{#1}{Reduction~\ref*{#1}}}
 \newrefformat{part}{\savehyperref{#1}{Part~\ref*{#1}}}

\newcommand{\Sref}[1]{\hyperref[#1]{\S\ref*{#1}}}

\usepackage{nicefrac}

\ifnum\usemicrotype=1
\usepackage{microtype}
\fi

\ifnum\showauthornotes=1
\else
\fi

\ifnum\showfixme=0
\fi

\usepackage{boxedminipage}

\newcommand{\iprod}[1]{\langle#1\rangle}

\newcommand{\Esymb}{\mathbb{E}}
\newcommand{\Psymb}{\mathbb{P}}

\DeclareMathOperator*{\E}{\Esymb}

\DeclareMathOperator*{\ProbOp}{\Psymb}

\renewcommand{\Pr}{\ProbOp}

\newcommand{\textparen}[1]{\text{(#1)}}

\ifx\because\undefined
\newcommand{\because}[1]{\textparen{because #1}}
\else
\renewcommand{\because}[1]{\textparen{because #1}}
\fi

\newcommand{\defeq}{\stackrel{\mathrm{def}}=}

\newcommand\bdot\bullet

\ifx\mathds\undefined 
\DeclareMathOperator{\Ind}{\mathbb{I}}
\else
\DeclareMathOperator{\Ind}{\mathds 1}}
\fi

\DeclareMathOperator{\polylog}{polylog}
\DeclareMathOperator{\supp}{supp}

\newcommand{\R}{\mathbb R}

\newcommand{\cC}{\mathcal C}
\newcommand{\cD}{\mathcal D}

\newcommand{\cS}{\mathcal S}

\renewcommand{\leq}{\leqslant}
\renewcommand{\le}{\leqslant}
\renewcommand{\geq}{\geqslant}
\renewcommand{\ge}{\geqslant}

\ifnum\showdraftbox=1
\newcommand{\draftbox}{\begin{center}
  \fbox{%
    \begin{minipage}{2in}%
      \begin{center}%
          \Large\textsc{Working Draft}\\%
        Please do not distribute%
      \end{center}%
    \end{minipage}%
  }%
\end{center}
\vspace{0.2cm}}
\else
\newcommand{\draftbox}{}
\fi

\let\epsilon=\varepsilon

\numberwithin{equation}{section}

\newcommand\MYcurrentlabel{xxx}

\newcommand{\MYstore}[2]{%
  \global\expandafter \def \csname MYMEMORY #1 \endcsname{#2}%
}

\newcommand{\MYload}[1]{%
  \csname MYMEMORY #1 \endcsname%
}

\newcommand{\MYnewlabel}[1]{%
  \renewcommand\MYcurrentlabel{#1}%
  \MYoldlabel{#1}%
}

\newcommand{\MYdummylabel}[1]{}

\newcommand{\torestate}[1]{%
  \let\MYoldlabel\label%
  \let\label\MYnewlabel%
  #1%
  \MYstore{\MYcurrentlabel}{#1}%
  \let\label\MYoldlabel%
}

\newcommand{\restatetheorem}[1]{%
  \let\MYoldlabel\label
  \let\label\MYdummylabel
  \begin{theorem*}[Restatement of \prettyref{#1}]
    \MYload{#1}
  \end{theorem*}
  \let\label\MYoldlabel
}

\newcommand{\restatelemma}[1]{%
  \let\MYoldlabel\label
  \let\label\MYdummylabel
  \begin{lemma*}[Restatement of \prettyref{#1}]
    \MYload{#1}
  \end{lemma*}
  \let\label\MYoldlabel
}

\newcommand{\restateprop}[1]{%
  \let\MYoldlabel\label
  \let\label\MYdummylabel
  \begin{proposition*}[Restatement of \prettyref{#1}]
    \MYload{#1}
  \end{proposition*}
  \let\label\MYoldlabel
}

\newcommand{\restatefact}[1]{%
  \let\MYoldlabel\label
  \let\label\MYdummylabel
  \begin{fact*}[Restatement of \prettyref{#1}]
    \MYload{#1}
  \end{fact*}
  \let\label\MYoldlabel
}

\newcommand{\restate}[1]{%
  \let\MYoldlabel\label
  \let\label\MYdummylabel
  \MYload{#1}
  \let\label\MYoldlabel
}

\newcommand{\addreferencesection}{
  \phantomsection
  \addcontentsline{toc}{section}{References}
}

\newcommand{\eps}{\epsilon}

\let\origparagraph\paragraph
\renewcommand{\paragraph}[1]{\origparagraph{#1.}}

\allowdisplaybreaks

\usepackage{paralist}

\usepackage{comment}

\usepackage{braket}


%% file: preamble/macros-custom.tex
\usepackage{relsize}
\usepackage[font=footnotesize]{caption}
\usepackage{appendix}

\DeclareUrlCommand\email{}

\DeclareMathOperator{\zo}{\{0,1\}}

\DeclareMathOperator*{\pE}{\widetilde{\mathbb E}}
\DeclareMathOperator*{\pPr}{\widetilde{\mathbb P}}
\DeclareMathOperator*{\pVar}{\widetilde{{\mathbb V}{\mathbb A}{\mathbb R}}}

\let\pref=\prettyref



%% file: preamble/ryan-macros.tex
\DeclarePairedDelimiter\parens{\lparen}{\rparen}
\DeclarePairedDelimiter\absryan{\lvert}{\rvert}

\DeclarePairedDelimiter\braces{\lbrace}{\rbrace}

\newcommand{\wt}[1]{\widetilde{#1}}
\newcommand{\wh}[1]{\widehat{#1}}
\newcommand{\avg}{\mathop{\mathrm{avg}}}

\newcommand{\littlesum}{\mathop{{\textstyle \sum}}}

\newcommand{\calC}{\mathcal{C}}
\newcommand{\calD}{\mathcal{D}}

\newcommand{\calH}{\mathcal{H}}

\newcommand{\calM}{\mathcal{M}}

\newcommand{\calP}{\mathcal{P}}

\newcommand{\calX}{\mathcal{X}}

\newcommand{\bw}{\boldsymbol{w}}
\newcommand{\bx}{{\boldsymbol{x}}}

\newcommand{\bC}{\boldsymbol{C}}

\newcommand{\bM}{\boldsymbol{M}}

\newcommand{\bU}{\boldsymbol{U}}

\newcommand{\OBJ}{\mathrm{OBJ}}

%% file: content/abstract.tex
Many previous Sum-of-Squares (SOS) lower bounds for CSPs had two deficiencies related to global constraints. First, they were not able to support a ``cardinality constraint'', as in, say, the Min-Bisection problem.  Second, while the pseudoexpectation of the objective function was shown to have some value~$\beta$, 
it did not necessarily actually ``satisfy'' the constraint ``objective~=~$\beta$''.  In this paper we show how to remedy both deficiencies in the case of random CSPs, by translating \emph{global} constraints into \emph{local} constraints. Using these ideas, we also show that degree-$\Omega(\sqrt{n})$ SOS does not provide a $(\frac{4}{3} - \eps)$-approximation for Min-Bisection, and degree-$\Omega(n)$ SOS does not provide a $(\frac{11}{12} +  \eps)$-approximation for Max-Bisection or a $(\frac{5}{4} - \eps)$-approximation for Min-Bisection.  No prior SOS lower bounds for these problems were known.

%% file: content/intro.tex
\section{Introduction}  \label{sec:intro}

Consider the task of refuting a random 3SAT instance with $n$ variables and $50n$ clauses; i.e., certifying that it's unsatisfiable (which it is, with very high probability).  There is no known $2^{o(n)}$-time algorithm for this problem.  An oft-cited piece of evidence for the exponential difficulty is the fact~\cite{Gri01,Sch08} that the very powerful Sum-of-Squares (SOS) SDP hierarchy fails to refute such random 3SAT instances in $2^{o(n)}$ time.  Colloquially, degree-$\Omega(n)$ SOS ``thinks'' that the random 3SAT instance is satisfiable (with high probability).

But consider the following method of refuting satisfiability of a random $50n$-clause CNF~$\phi$:
\begin{quotation}
    For all $k \in \{0, 1, 2, \dots, n\}$,

    \qquad refute ``$\phi$ is satisfiable by an assignment of Hamming weight~$k$''.
\end{quotation}
Could it be that $O(1)$-degree SOS succeeds in refuting random 3SAT instances in this manner?  It seems highly unlikely, but prior to this work the possibility could not be ruled out.

\paragraph{SOS lower bounds with Hamming weight constraints}  Recall that the known SOS lower bounds for random 3SAT are actually stronger: they show degree-$\Omega(n)$ SOS thinks that random 3SAT instances are satisfiable even as \emph{as 3XOR} (i.e., with every clause having an odd number of true literals).  Hamming weight calculations are quite natural in the context of random 3XOR; indeed Grigoriev, Hirsch, and Pasechnik~\cite{GHP02} showed that the \emph{dynamic} degree-$5$ SOS proof system \emph{can} refute random 3XOR instances by using integer counting techniques.  Thus the above ``refute solutions at each Hamming weight'' strategy seems quite natural in the context of random CSPs.

In 2012, Yuan Zhou raised the question of proving strong SOS lower bounds for random 3XOR instances together with a global cardinality constraint such as $\sum_i x_i = \frac{n}{2}$.  This would rule out the above refutation strategy.  It is also a natural SOS challenge, seemingly combining the two strong SOS results known prior to 2012 --- the lower bound for random 3XOR due to Grigoriev and Schoenebeck~\cite{Gri01,Sch08} and the lower bound for Knapsack due to Grigoriev~\cite{Gri01a}.

One may ask why the Grigoriev--Schoenebeck SOS lower bound doesn't already satisfy $\sum_i x_i = \frac{n}{2}$.  The difficulty is connected to the meaning of the word ``satisfy''.  One should think of the SOS Method as trying to find not just a satisfying assignment to a CSP, but more generally a \emph{distribution} on satisfying assignments.  The SOS algorithm finds a ``degree-$d$ pseudodistribution'' on satisfying assignments in $n^{O(d)}$ time, provided one exists; roughly speaking, this means an object that ``looks like'' a distribution on satisfying assignment to all tests that are squared polynomials of degree at most~$d$.  For a random 3XOR instance with $n$ variables and $O(n)$ constraints, the Grigoriev--Schoenebeck degree-$\Omega(n)$ pseudodistribution indeed claims to have 100\% of its probability mass on satisfying assignments.  Furthermore, its assignments claim to give probability 50\% to each of $x_i = 0$ and $x_i = 1$ for all~$i$; in other words, the ``pseudoexpectation'' of $x_i$ is $\frac12$, and therefore the pseudoexpectation of $\sum_i x_i$ is~$\frac{n}{2}$.  However, this \emph{doesn't} mean that the pseudodistribution ``satisfies'' the hard constraint $\sum_i x_i = \frac{n}{2}$, in the usual sense by which one speaks of an SOS solution satisfying constraints.  To actually ``satisfy'' this constraint, the expression $\sum_i x_i$ must have \emph{pseudovariance zero}; i.e., SOS must not only ``think'' it knows a distribution on 3XOR-satisfying assignments which has $\sum_i x_i = \frac{n}{2}$ on average, it must think that \emph{all} of these satisfying assignments have $\sum_{i} x_i$ exactly~$\frac{n}{2}$.

In this work we show how to upgrade any SOS lower bound for random CSPs based on $t$-wise uniformity so as to include the hard cardinality constraint $\sum_i x_i = \frac{n}{2}$ (or indeed $\sum_i x_i = \frac{n}{2} + k$ for any $|k| = O(\sqrt{n})$).\footnote{We also show in \pref{app:imbal} that this is not too far from tight, in the sense that it is easier to refute XOR with Hamming weight constraints that are too imbalanced (if $k = \omega(n^{1/4})$).}  The idea is conceptually simple: just add a matching of 2XOR constraints, $x_{2i-1} \neq x_{2i}$ for all $1 \leq i \leq \frac{n}{2}$.

\paragraph{SOS lower bounds with exact objective constraints}  A random 3AND CSP with $n$ variables and~$m = \alpha n$ constraints (each an AND of~$3$ random literals) will have objective value $\frac18 + \epsilon$ with high probability, for $\epsilon$ arbitrarily small as a function of $\alpha$; i.e., the best assignment will satisfy at most $(\frac18 + \epsilon) m$ constraints.  On the other hand, it's not too hard to show that the Grigoriev--Schoenebeck degree-$\Omega(n)$ pseudodistribution will give the objective function a pseudoexpectation of~$\frac14 \pm o(1)$.  (Roughly speaking, for almost all 3AND constraints, the SOS pseudodistribution will think it can obtain probability~$\frac14$ on each of the 3XOR-satisfying assignments, and one of these, namely $(1,1,1)$, satisfies 3AND.)  Thus it would appear that degree-$\Omega(n)$ SOS  has an integrality gap of factor $2-\epsilon $ on random 3AND instances.

But is this misleading?  Suppose we solved the SOS SDP and it reported a solution with pseudoexpectation~$\frac14$.  We might then ``double-check'' by re-running the SDP, together with an additional ``equality constraint'' specifying that the number of satisfied 3AND constraints is indeed~$\frac14 m$.\footnote{Actually, it was recently observed that it is not clear we can definitely solve the associated SDP exactly~\cite{OD17,BW17}.  This does not affect the status of our lower bounds.}  As far as we know now, this run could return ``infeasible'', actually \emph{refuting} the possibility of $\frac14 m$ constraints being satisfiable!  Again, the issue is that under the Grigoriev--Schoenebeck SOS pseudodistribution, the objective function will have a pseudoexpectation like~$\frac14$, but will also have \emph{nonzero pseudovariance}.

We show how to fix this issue --- i.e., have the objective constraint be \emph{exactly} SOS-satisfied --- in the context of any SOS lower bound for random CSPs based on $t$-wise uniformity.  Here we briefly express the idea of our solution, in the specific case of 3AND:  We show that one can design a probability distribution~$\theta$ on $r \times 4$ Boolean matrices such that two properties hold: (i)~$\theta$ is $2$-wise uniform; (ii)~for \emph{every} outcome in the support of~$\theta$, \emph{exactly} a $\frac14 - \eps_r$ fraction of the~$r$ rows satisfy 3AND, where $\eps_r$ is an explicit positive constant depending on~$r$ that tends to~$0$ as $r$ grows.  We then use recent work~\cite{KMOW17} on constructing SOS lower bounds from $t$-wise uniform distributions to show that degree-$\Omega(n)$ SOS thinks it can ``weakly satisfy'' a random ``distributional CSP'' in which each constraint specifies that a random $4r$-tuple of variables should be distributed according to~$\theta$.  By ``weak satisfaction'', we mean that SOS will at least think it can get a local distribution on each $4r$-tuple whose support is contained within~$\theta$'s support (and therefore always having exactly a $\frac14 - \eps_r$ fraction of rows satisfying 3AND).  Now viewing each such tuple as the conjunction of~$r$ (random) 3AND constraints, we get that the SOS solution thinks it satisfies exactly a $\frac14 - \eps_r$ fraction of these constraints.

\paragraph{Further consequences}  Via our first result --- satisfying global cardinality constraints --- we open up the possibility of establishing SOS lower bounds for natural problems like Min- and Max-Bisection (by performing reductions within SOS, as in \cite{DBLP:conf/stoc/Tulsiani09}).  Previously, no such SOS integrality gaps were known (Guruswami, Sinop, and Zhou~\cite{GSZ14} had given an SOS integrality gap approaching~$\frac{11}{10}$ for the Balanced-Separator problem, which is like Min-Bisection but without a hard bisection constraint.)  Under assumptions like $\mathsf{NP} \not \subseteq \bigcap_{\eps > 0} \mathsf{TIME}(2^{n^\eps})$, some hardness results were previously known: no PTAS for Min-Bisection (due to Khot~\cite{Kho06}) and factor $\frac{15}{16} +\eps$ hardness for Max-Bisection (due to Holmerin and Khot~\cite{Kho06}, improving on the factor $\frac{16}{17} + \eps$ $\mathsf{NP}$-hardness known for Max-Cut).  However, in the context of SOS lower bounds, it makes sense to shoot for more: namely, hardness factors that are known subject to Feige's R3SAT Hypothesis~\cite{Fei02} (and similar hypotheses for random CSPs).

Feige himself~\cite{Fei02} showed factor $\frac43 - \eps$ hardness for Min-Bisection under his hypothesis (with a quadratic size blowup).  Also, it's possible to show factor $\frac{11}{12} + \eps$ hardness for Max-Bisection (with linear size blowup) under Feige's Hypothesis for 4XOR; this is arguably ``folklore'', via the gadget techniques of Trevisan et al.~\cite{TSSW00} (see also~\cite{Has01,OW12}).  We are able to convert both of these results to SOS lower bounds, showing that degree-$\Omega(\sqrt{n})$ SOS fails to  $(\frac43 - \eps)$-approximate Min-Bisection, and degree-$\Omega(n)$ SOS fails to  $(\frac54 + \eps)$-approximate Min-Bisection.
Our proof of the latter can also be modified to show that degree-$\Omega(n)$ SOS fails to $(\frac{11}{12} + \eps)$-approximate Max-Bisection.

It is worth pointing out that the benefit of our second main result, the ability to enforce objective equality constraints exactly, also arises in these SOS Bisection lower bounds.  For example, the $(\frac{4}{3} - \eps)$-hardness for Min-Bisection is a kind of gadget reduction from random 3AND CSPs; showing that the ``good cut'' in the completeness case is an exact bisection relies on the ``good assignment''  in the 3AND instance satisfying exactly a~$\frac14$ fraction of constraints.

\subsection{Statement of main theorems}\label{sec:mainthms}
Recent work~\cite{BCK15,KMOW17} has established a general framework for showing lower bounds for SOS on random CSPs, using the idea of $t$-wise uniformity.  The following is a fairly general example of what's known:
\setcounter{theorem}{-1}
\begin{theorem}                                     \label{thm:kmow-orig}
    (\cite{KMOW17}.)  Let $P : \{0,1\}^k \to \{0,1\}$ be a predicate, and suppose there is a $(t-1)$-wise uniform distribution~$\nu$ on $\{0,1\}^k$ with $\E_\nu[P] = \beta$.  Consider a random $n$-variable, $m = \Delta n$-constraint instance of CSP$(P^{\pm})$, meaning that each constraint is~$P$ applied to~$k$ randomly chosen literals.  Then with high probability, there is a degree-$\Omega\left(\frac{n}{\Delta^{2/(t-2)}\log \Delta}\right)$ SOS pseuodexpectation~$\pE[\cdot]$ with the following property:

    \textbf{Case 1,} $\beta = 1$.  In this case, $\pE[\cdot]$ satisfies all the CSP constraints as identities.

    \textbf{Case 2,} $\beta < 1$.  In this case, $\pE[\OBJ(x)] = \beta \pm o(1)$, where $\OBJ(x)$ denotes the objective value of the CSP.
\end{theorem}
For example, the case of random 3SAT described in the previous section corresponds to $P = \mathrm{OR}_3$, $t = 3$, $\nu$~being the uniform distribution on triples satisfying $\mathrm{XOR}_3$, $\beta = 1$, and $\Delta = 50$; the case of random 3AND has the same $t$, $\nu$, and $\Delta$, but $P = \mathrm{AND}_3$ and $\beta = \frac14$. \\

Our main theorems are now as follows:
\begin{theorem}                                     \label{thm:main1}
    In the $\beta = 1$ case of \pref{thm:kmow-orig}, one can additionally get the pseudodistribution~$\pE$ to satisfy (with pseudovariance zero) the global bisection constraint $\sum_{i=1}^n x_i = \frac{n}{2}$ (assuming~$n$ even).  More generally, for any integer $B \in [\frac{n}{2} - O(\sqrt{n}), \frac{n}{2} + O(\sqrt{n})]$, we can ensure the pseudodistribution satisfies the global Hamming weight constraint $\sum_{i=1}^n x_i = B$.
\end{theorem}

\begin{theorem}                                     \label{thm:main2}
    In the $\beta < 1$ case of \pref{thm:kmow-orig}, there exists a sequence of positive constants $\eps_r$ with $\eps_r \to 0$ such that for a random* $n$-variable, $m = \Delta n$-constraint instance of CSP$(P^\pm)$, with high probability there is a degree-$\Omega_r\left(\frac{n}{\Delta^{2/(t-2)}\log \Delta}\right)$ SOS pseudodistribution $\pE$ which satisfies (with pseudovariance zero) the hard constraint ``$\OBJ(x) = \beta - \eps_r$''.  Furthermore, we can also obtain cardinality constraints as in \pref{thm:main1}.
\end{theorem}
	In \pref{app:imbal}, we show that \pref{thm:main1} is not too far from tight, by demonstrating that random $k$-XOR instances become easier to refute when one imposes an imbalanced Hamming weight constraint $\sum_{i=1}^n x_i = \frac{n}{2} \pm \omega(n^{1/4})$.\footnote{We conjecture that \pref{thm:main1} is tight, and that Hamming weight constraints with imbalance $\omega(n^{1/2})$ already make $k$-XOR easier to refute.}

\begin{remark} \label{rem:random-star}
    In the above theorem we have written ``random*'' with an asterisk because the random instance is not drawn precisely in the standard way.  Rather, it is obtained by choosing $m/r$ groups of random constraints, where in each group we fix a literal pattern and then choose~$r$ nonoverlapping constraints with this pattern.  This technicality is an artifact of our proof; it seems likely that it is unnecessary.  
    Indeed, it is possible that these two distributions on random hypergraphs are simply $o(1)$-close in total variation distance, at least when $m = O(n)$.\footnote{Thanks to Svante Janson for some observations in the direction of showing this.}  In any case, by alternate means (including the techniques from \pref{thm:main1}) we are able to show the following alternative result in \pref{sec:sparse-ex}: When $m = o(n^{1.5})$, with high probability a purely random instance of CSP$(P^\pm)$ has an SOS pseudodistribution of the stated degree that exactly satisfies $\OBJ(x) = \beta - \eps$ for some $\eps > 0$ that can be made arbitrarily small.
\end{remark}

\begin{remark}
    Our proof of \pref{thm:main2} only relies on the ``Case~1, $\beta = 1$'' part of~\cite{KMOW17}'s \pref{thm:kmow-orig}.  In fact, our \pref{thm:main2} can actually be used to effectively \emph{deduce} ``Case~2, $\beta < 1$'' from ``Case~1, $\beta = 1$'' in \pref{thm:kmow-orig}.  This is of interest because~\cite{KMOW17}'s argument for Case~2 was not a black-box reduction from Case~1, but instead involved verifying a more technical expansion property in random graphs, as well as slightly reworking the proof of Case~2.
\end{remark}

Finally, we obtain the following theorems concerning Bisection problems:

\begin{theorem}\label{thm:apps}
    For the Max-Bisection problem in a graph on $n$ vertices, for $d = \Omega(n)$, the degree-$d$ Sum-of-Squares Method cannot obtain an approximation factor better than $\frac{11}{12}-\eps$ for any constant $\eps > 0$.

    For the Min-Bisection problem, for $d' = \Omega(\sqrt{n})$, the degree-$d'$ SOS Method cannot obtain an approximation factor better than $\frac{4}{3} - \eps$, and for $d = \Omega(n)$ the degree-$d$ SOS Method cannot obtain an approximation factor better than $\frac{5}{4}-\eps$.
\end{theorem}

\subsection*{Organization of this paper}
In \pref{sec:prelims}, we provide some preliminaries and technical context for the study of CSPs and SOS.
In \pref{sec:csps-with-matching} (and in \pref{app:kmow} and \pref{app:expansion}), we extend the results of \cite{KMOW17} to obtain lower bounds for CSPs with global cardinality constraints, proving \pref{thm:main1}.
\pref{sec:localdist} shows how to construct local distributions over assignments to groups of disjoint predicates so that the number of satisfied constraints is always exactly the same, and \pref{sec:splitting} shows how to use such distributions to prove \pref{thm:main2}.
In \pref{sec:splitting}, one can also find a discussion of random vs. random* CSPs.
Finally, \pref{sec:feige} and \pref{sec:bisect} contain our applications to Min- and Max-Bisection.
We wrap up with some concluding remarks and future directions in \pref{sec:conclusion}.
In \pref{app:imbal}, we observe that \pref{thm:main1} is not too far from tight, as random $k$-XOR formulas in the presence of more dramatically imbalanced Hamming weight constraints are easier to refute.

%% file: content/prelims.tex
\section{Preliminaries}  \label{sec:prelims}

\newcommand{\inst}{\calH}
\newcommand{\con}{h}
\newcommand{\cdist}{\nu}
\newcommand{\gdist}{\mu}

\paragraph{CSPs} A constraint satisfaction problem (CSP) is defined by an \emph{alphabet} $\Omega$ (usually~$\{0,1\}$ or $\{\pm 1\}$ in this paper) and a collection $\calP$ of \emph{predicates}, each predicate being some $P : \Omega^k \to \{0,1\}$ (with different $P$'s possibly having different \emph{arities},~$k$).  An \emph{instance} $\inst$ consists of a set~$V$ of $n$ variables, as well as~$m$ constraints. Each constraint~$\con$ consists of a \emph{scope}~$S$ and a predicate $P \in \calP$, where $S$ is a tuple of $k$ distinct variables, $k$ being the arity of~$P$.  An \emph{assignment} gives a value $x_i \in \Omega$ to the $i$th variable; it \emph{satisfies} constraint~$\con = (S,P)$ if $P(x_{S_1}, \dots, x_{S_k}) = 1$.  We may sometimes write this as $P(x_S) = 1$ for brevity. The associated \emph{objective value} is the fraction of satisfied constraints,
\[
    \OBJ(x) = \avg_{\con = (S,P) \in \inst} \braces*{P(x_S)}.
\]
Sometimes we are concerned with CSPs of the following type: the alphabet~$\Omega = \{\pm 1\}$ is Boolean, there is a single predicate~$P : \Omega^k \to \{0,1\}$ (e.g., $P = \mathrm{OR}_3$, the $3$-ary Boolean OR predicate), and the predicate set~$\calP$ consists of all~$2^k$ versions of~$P$ in which inputs may be negated.  We refer to this scenario as $P$-CSP with \emph{literals}, denoted CSP$(P^\pm)$.  For example, the case of $P = \mathrm{OR}_3$ is the classic ``3SAT'' CSP.

\paragraph{Distributional CSPs} A \emph{distributional CSP} is one where, rather than having a predicate associated with each scope, we have a probability distribution.  More precisely, each \emph{distributional constraint} $\con = (S,\cdist)$ now consists of a scope~$S$ of some arity~$k$, as well as a probability distribution $\cdist$ on~$\Omega^k$.  The optimization task involves finding a \emph{global probability distribution}~$\gdist$ on assignments.  We say that $\gdist$ \emph{satisfies} constraint~$\con = (S,\cdist)$ if the marginal $\gdist|_S$ of~$\gdist$ on~$S$ is equal to~$\cdist$; we say the distributional CSP is \emph{satisfiable} if there is a~$\gdist$ satisfying all constraints.

 We may also say that $\gdist$ \emph{weakly satisfies} $\con = (S,\cdist)$ if $\supp(\gdist|_S) \subseteq \supp(\cdist)$.   A ``usual'' (\emph{predicative}, i.e., non-distributional) CSP can be viewed as a distributional CSP as follows:  For each predicate~$P$, select any distribution~$\cdist_P$ whose support is exactly the satisfying assignments to~$P$; then the existence of a global assignment in the predicative CSP of objective value~$\beta$ is equivalent to the existence of a global probability distribution~$\gdist$ that \emph{weakly} satisfies a~$\beta$ fraction of constraints.

\paragraph{Random CSPs}  We are frequently concerned with CSPs chosen uniformly at random.  Given a predicate set~$\calP$, a random CSP with $n$ variables and $m$ constraints is chosen as follows:  For each constraint we first choose a random~$P \in \calP$.  Supposing it has arity~$k$, we then choose a uniformly random length-$k$ scope~$S$ from the~$n$ variables, and impose the constraint~$(S,P)$.  We can similarly define a random distributional CSP given a collection~$\calD$ of distributions~$\cdist$.  We remark that our choice of having exactly~$m$ constraints is not really essential, and not much would change if we had, e.g., a Poisson$(m)$ number of random constraints, or if we chose each possible constraint independently with probability such that~$m$ constraints are expected.

\paragraph{SOS} The SOS Method \cite{BS16} can be thought of as an algorithmic technique for finding upper bounds on the best objective value achievable in a predicative or distributional CSP.  For example, in a random 3SAT instance with~$m = 50n$, it is very likely that every assignment~$x$ has $\OBJ(x) \leq \frac78 + o(1)$; ideally, the SOS Method could certify this, or could at least certify unsatisfiability, meaning an upper bound of $\OBJ(x) < 1$ for all assignments.  The SOS Method has a tunable \emph{degree} parameter~$d$; increasing~$d$ increases the effectiveness of the method, but also its run-time, which is essentially~$n^{O(d)}$ (though see~\cite{OD17,BW17} for a more precise discussion).   In this work we are only concerned with showing negative results for the power of SOS.  Showing that degree-$d$ SOS fails to certify a good upper bound on the maximum objective value is equivalent to showing that a \emph{degree-$d$ pseudodistribution} exists under which the objective function has a large \emph{pseudoexpectation}.  We define these terms now.

For simplicity we restrict attention to CSPs with Boolean alphabet (either $\Omega = \{0,1\}$ or $\Omega = \{\pm 1\}$), although it straightforward to extend the definitions for larger alphabets.\footnote{Specifically, for each variable~$x$ and each alphabet element~$a \in \Omega$, one introduces an indeterminate called~$1_{x=a}$ that is constrained as a $\{0,1\}$ value and is interpreted as the indicator of whether~$x$ is assigned~$a$.}  The SOS method introduces \emph{indeterminates} $X_1, \dots, X_n$ associated to the CSP variables; intuitively, one thinks of them as standing for the outcome of a global assignment chosen from a supposed probability distribution on assignments.  An associated \emph{degree-$d$ pseudoexpectation} is a real-valued linear map~$\pE$ on $\R_{\leq d}[X_1, \dots, X_n]$ (the space of formal polynomials in~$X_1, \dots, X_n$ of degree at most~$d$) satisfying three properties:
\vspace{5pt}
\begin{compactenum}
    \item $\pE[\text{multilin}(Q(X))] = \pE[Q(X)]$; here $\text{multilin}(Q(X))$ refers to the \emph{multilinearization} of $Q(X)$, meaning the reduction mod~$X_i^2 = X_i$ (in case $\Omega = \{0,1\}$) or mod~$X_i^2 = 1$ (in case $\Omega = \{\pm 1\}$).
    \item $\pE[1] = 1$;
    \item  \label{item:psd} $\pE[Q(X)^2] \geq 0$ whenever $\deg(Q) \leq d/2$.
\end{compactenum}
\vspace{5pt}
We tend to think of the first condition (as well as the linearity of~$\pE$) as being ``syntactically'' enforced; i.e., given $\pE$'s values on the multilinear monomials, its value on all polynomials is determined through multilinearization and linearity.  It is not hard to show that every pseudoexpectation~$\pE$ arises from a \emph{signed probability distribution}~$\mu$; i.e., a (possibly negative) function $\mu : \Omega^n \to \R$ with $\sum_x \mu(x) = 1$.  We call this the associated \emph{pseudodistribution}.  Intuitively, we think of a degree-$d$ pseudodistribution as a ``supposed'' distribution on global assignments, which at least passes the tests in \pref{item:psd} above.

Given a CSP instance~$\calH$, if there is a degree-$d$ pseudodistribution with $\pE[\OBJ(X)] \geq \beta$, this means that the degree-$d$ SOS Method \emph{fails} to certify an upper bound of $\OBJ(x) < \beta$ for the CSP.  Informally, we say that degree-$d$ SOS ``thinks'' that there is a distribution on assignments under which the average objective value is at least~$\beta$.  Similarly, given a distributional CSP~$\calH$, if there is a degree-$d$ pseudodistribution in which $\pPr[X_S = (a_1, \dots, a_k)] = \cdist(a_1, \dots, a_k)$ for all constraints $\con = (S,\cdist)$, we say that degree-$d$ SOS ``thinks'' that~$\calH$ is fully satisfiable.  Here $\pPr[X_S = (a_1, \dots, a_k)]$ means $\pE[1_{X_S = (a_1, \dots, a_k)}]$, where $1_{X_S = (a_1, \dots, a_k)}$ denotes the natural arithmetization of the $0$-$1$ indicator as a degree-$k$ multilinear polynomial.

\paragraph{Satisfaction of identities in SOS}  Formally speaking, one says that a degree-$d$ pseudodistribution \emph{satisfies an identity~$Q(X) = b$} if $\pE[(Q(X) - b) R(X)] = 0$ for all polynomials $R(X)$ of degree at most $d - \deg(Q)$.  Note that this is \emph{stronger} than simply requiring $\pE[Q(X)] = b$ (the $R \equiv 1$ case).  A great deal of this paper is concerned with precisely this distinction; it may be relatively easy to come up with a degree-$d$ pseudodistribution over $\{0,1\}^n$ satisfying, say, $\pE[\sum_i X_i] = \frac{n}{2}$, but much harder to find one that ``satisfies the identity $\sum_i X_i = \frac{n}{2}$''.  The terminology here is a little unfortunate; we will try to ameliorate things by introducing the following stronger phrase:
\begin{definition}
    We say that a degree-$d$ pseudodistribution satisfies identity $Q(X) = b$ \emph{with pseudovariance zero} if we have both $\pE[Q(X)] = b$ and \emph{also}
    \[
        \pVar[Q(X) - b] = \pE[Q(X)^2] - b^2 = 0.
    \]
As is shown in~\cite[Lemma~3.5~(SOS Cauchy-Schwarz)]{BBH+12}, this condition is equivalent to the pseudodistribution ``satisfying the identity~$Q(X) = b$'' for $Q$ of degree up to $d/2$.\footnote{In this paper we are flexible when it comes to constant factors in the degree. For this reason we need not worry about this factor-$2$ loss in the degree, as a degree-$2d$ pseudoexpectation which satisfies an identity with pseudovariance $0$ automatically gives a degree-$d$ pseudoexpectation which satisfies the identity exactly.}
\end{definition}
Intuitively, in this situation degree-$d$ SOS not only ``thinks'' that it knows a distribution on assignments~$x$ under which~$Q(x)$ has expectation~$b$, it further thinks that \emph{every} outcome~$x$ in the support of its supposed assignment has $Q(x) = b$.

%% file: content/expansion.tex
\section{Random CSPs with Hamming weight constraints}     \label{sec:csps-with-matching}

In this section we will prove \pref{thm:main1}, which extends the known random CSP lower bounds (as in \pref{thm:kmow-orig}) to CSPs with a hard Hamming weight constraint on the variable assignment.
Our lower bounds build on the machinery developed in~\cite{KMOW17}, so we first recall the setup and main theorems from that work; then we will extend the result from that work, and finally prove \pref{thm:main1}.

\subsection{Hypergraph expansion and prior SOS lower bounds for random CSPs}
The paper~\cite{KMOW17} works in the general setting of distributional CSPs with an upper bound of~$K$ on all constraint arities.  An instance is thought of as a ``factor graph''~$G$: a bipartite graph with $n$ variable-vertices, $m$ constraint-vertices, and edges joining a constraint-vertex to the variable-vertices in its scope.  More precisely, the neighborhood $N(h)$ of each constraint-vertex $h$ is defined to be an \emph{ordered} tuple of $k_h$ variable-vertices. 
We write $\nu_h$ for the local probability distribution on $\Omega^{N(h)}$ associated to constraint~$h$.  In~\cite{KMOW17}, each~$\nu_h$ is assumed to be a $(\tau-1)$-wise uniform distribution, where $\tau$ is a global integer parameter satisfying $3 \leq \tau \leq K$.  Finally, the graph $G$ is assumed to satisfy a certain high-expansion condition (discussed in \pref{app:kmow}) called the ``Plausibility Assumption'' involving two parameters $0 < \zeta < 1$ and $1 \leq \CONSMALL \leq n/2$, assumed to satisfy $K \leq \zeta \cdot \CONSMALL$. 
 In this case, the main theorem of~\cite{KMOW17} is that there is a SOS-pseudodistribution of degree $\frac13 \zeta \cdot \CONSMALL$ that weakly satisfies all constraints.  

In~\cite{KMOW17} it is assumed that all constraint distributions $\nu_h$ have the \emph{same} level of uniformity, namely $(\tau-1)$-wise uniformity, $\tau \geq 3$.  In this work, in order to incorporate Hamming weight constraints on the assignment, we would like to consider the possibility that different constraint distributions have different levels of uniformity.  To that end, suppose that each $\nu_h$ is $(t_h-1)$-wise uniform, where the $t_h$'s are various integers.  Slightly more broadly than~\cite{KMOW17}, we allow $1 \leq t_h \leq k_h+1$ for all~$h$, and we allow the constraints to have arity $k_h$ as low as~$1$.

In \pref{app:kmow}, we examine how these assumptions affect the proofs in~\cite{KMOW17}.  The upshot is \pref{thm:good-if-expander} below.  Before we give the theorem, we briefly introduce some notation and comments: A ``constraint-induced'' subgraph~$H$ is a subgraph of the factor graph~$G$ given by choosing some set of constraints~$C$, as well as all edges and constraint-vertices adjacent to~$C$.  We write $c(H)$ for the number of constraints in~$H$, $e(H)$ for the number of edges, $v(H)$ for the number of variable-vertices, and $T(H) = \sum_{h \in \cons(H)} t_h$.  To reduce to \pref{eqn:plaus} in the following theorem, we use the observation from \pref{def:income-saturation} in \pref{app:kmow} that adding edges to a subgraph to make it constraint-induced can only decrease ``income''.
For notational simplicity we have also adjusted the parameters $\zeta$ and $\CONSMALL$ by factors of~$2$.
\begin{theorem}                                     \label{thm:good-if-expander}
    (Essentially from~\cite{KMOW17}.) Let $0 < \zeta < 1$, $\CONSMALL \leq n$ and assume all
    constraint-vertices in~$G$ have arity at most $\zeta \cdot \CONSMALL$.

    Suppose that for every set of nonempty constraint-induced subgraph~$H$ with $c(H) \leq \CONSMALL$, it holds that
    \begin{equation}    \label{eqn:plaus}
        v(H) \geq e(H) - \frac{T(H)}{2} + \zeta c(H).
    \end{equation}
    Then there is an SOS-pseudodistribution of degree $\frac13 \zeta \cdot \CONSMALL$ that weakly satisfies all constraints.
\end{theorem}

There are a lot of parameters in the above theorem, and our goal is not to derive the most general possible quantitative result. Instead we'll simply work out some of the basic consequences.

A basic setting treated in~\cite{KMOW17}, relevant for \pref{thm:kmow-orig}, is the following.  For a fixed small~$t$ we choose a random CSP with~$n$ variables and $\Delta n$ constraints, with each constraint supporting a $(t-1)$-wise uniform distribution. E.g., in random $3$SAT, $t =3$.
Then if
\[
    \Delta = \text{const} \cdot \parens*{\frac{n}{\CONSMALL}}^{\frac{t}{2} - 1 - \zeta}
\]
for a sufficiently small positive constant, it is shown in~\cite{KMOW17} that the main condition \pref{eqn:plaus} indeed holds with high probability. Choosing, say, $\zeta = \frac{1}{\log n}$ and $\CONSMALL = \polylog(n)$, we see that, with high probability, we will have weakly satisfying pseudodistributions of degree $\polylog(n)$ even when $\Delta = \wt{\Theta}(n^{t/2 - 1})$.

In fact, it's possible to show that we have such pseudoexpectations when there are, simultaneously, $n^{1.5}/\polylog(n)$ $2$-wise-supporting constraints, \emph{and} $n^{2}/\polylog(n)$ $3$-wise-supporting constraints, \emph{and} $n^{2.5}/\polylog(n)$ $4$-wise-supporting constraints, \dots \emph{and also} $n/\polylog(n)$ $1$-wise-supporting constraints, and $n^{.5}/\polylog(n)$ $0$-wise-supporting constraints.

\subsection{Expansion in the presence of matching and unary constraints}
However, if we want to impose a cardinality constraint by way of adding $1$-wise independent $2$-ary $\neq$ constraints, then $n/\polylog(n)$ such constraints will not suffice.
Indeed, what we would like to now show is that if the $1$-wise-supporting constraints are carefully chosen to not overlap, we can add a full, linear-sized ``matching'' of them without compromising the lower bound.  Then, when $\Omega = \{0,1\}$, we can impose the $1$-wise-uniform constraints $x_1 \oplus x_2 = 1$, $x_3 \oplus x_4 = 1$, \dots, $x_{n-1} \oplus x_{n} = 1$ and thereby force the pseudoexpectation to satisfy the global constraint $\sum_i x_i = \frac{n}{2}$.

\begin{theorem}                                     \label{thm:random-graph}
    Fix a uniformity parameter $3 \leq t \leq O(1)$, an arity $k \leq O(1)$, a number $U = O(\sqrt{n})$ of ``unary'' constraints, and a small failure probability $0 < p < 1/2$.  Assume also that $\zeta \leq 1/2$.

    Suppose we form a random factor graph  with $n$ variable-vertices and $\Delta n$ constraint-vertices~$\calC$ of arity~$k$; assume each constraint-vertex is equipped with an associated $(t-1)$-wise uniform distribution.

    Furthermore, suppose we add in two sets $\calM_1$,  $\calM_2$ of nonrandom, nonoverlapping constraints, whose associated variable vertices partition~$[n]$.  The ``unary'' constraints of $\calM_1$ should satisfy $|\calM_1| \leq U$ and have an associated $0$-wise uniform distribution; the ``matching'' constraints of $\calM_2$  should be of constant arity and have an associated $1$-wise uniform distribution.

    Then provided
    \[
        \Delta \leq \text{const} \cdot p \cdot \parens*{\frac{n}{\CONSMALL}}^{\frac{t}{2}-1-\zeta'}
    \]
    for a sufficiently small universal constant, and $\zeta' = (k+1)\zeta$,  the expansion condition in \pref{thm:good-if-expander} holds except with probability at most~$p$.
\end{theorem}
The proof of \pref{thm:random-graph}, appearing in \pref{app:expansion} uses standard combinatorial techniques for verifying the expansion of random graphs. Due to the fact that the unary and matching constraints are deterministic, we must augment these standard techniques with some straightforward case analysis. In fact, we only prove the theorem under the assumption that the constraints of~$\calM_2$ have arity~$2$; the more case of general constant arity is a slight elaboration that we omit.

\subsection{Lower bound for CSPs with Hamming weight constraints}
As in~\cite{KMOW17}, we observe that for a given $\Delta \geq 10$, a good choice for $\zeta$ is $\frac{1}{\log \Delta}$.  This yields the following corollary, which we will show implies \pref{thm:main1}:
\begin{theorem}                                     \label{thm:random-graph}
   Let $\nu$ be a $(t-1)$-wise uniform distribution on $\{0,1\}^k$.  Consider a random $n$-variable, $m = \Delta n$-constraint $k$-ary distributional CSP, in which each constraint distribution is~$\nu$ up to a negation pattern in the~$k$ inputs. (All such ``reorientations'' are still $(t-1)$-wise uniform.)  Suppose we also impose the following nonrandom distributional constraints:
    \smallskip
    \begin{compactitem}
        \item The $0$-wise uniform  constraints $x_1 = b_1, x_2 = b_2, \dots, x_U = b_U$, for some string $b \in \{0,1\}^U$ with $U = O(\sqrt{n})$;
        \item The $1$-wise uniform constraint that $(x_{U+1},x_{U+2})$ is uniform on~$\{(0,1), (1,0)\}$, and similarly for the pairs $(x_{U+3},x_{U+4})$, \dots, $(x_{n-1},x_n)$.
    \end{compactitem}
    \smallskip
   Then with high probability, there is an SOS-pseudodistribution of degree~$D = \Omega\parens*{\frac{n}{\Delta^{2/(t-2)} \log \Delta}}$ that weakly satisfies all constraints.
\end{theorem}

Let us now see why this implies \pref{thm:main1}.  In this ``$\beta = 1$'' scenario, we have a $(t-1)$-wise uniform~$\nu$ supported on the satisfying assignments for~$P$.  Whenever the random CSP($P^\pm$) instance has a $P$-constraint with a particular literal pattern, we impose the analogous $\nu$-constraint with equivalent negation pattern.  Now the SOS-pseudodistribution~$\pE[\cdot]$ promised by \pref{thm:random-graph} weakly satisfies all these $\nu$-constraints, and hence satisfies all the~$P$-constraints.  Furthermore, it also has $\pE[x_i] = b_i \in \{0,1\}$ for all $i \leq U$, and $\pE[x_i(1-x_{i+1})] + \pE[(1-x_i)x_{i+1}] = 1$ for all pairs $(i,i+1) =(U+1,U+2), \dots, (n-1,n)$, by weak satisfaction.  Notice that the latter implies
\[
    \pE[(x_i + x_{i+1} - 1)^2] = 1 - \pE[x_i] -\pE[y_i]+2\pE[x_i x_{i+1}] = 0,
\]
and hence the SOS solution satisfies $x_i + x_{i+1} = 1$ with pseudovariance zero.  Similarly (and easier), it satisfies the identity $x_i = b_i$ for all $i \leq U$.  It now follows that the pseudodistribution satisfies the identity $\sum_{i=1}^n x_i = \frac{n}{2} + (|b| - \frac{U}{2})$ with pseudovariance zero, and this completes the proof of \pref{thm:main1}, because we can take any $|b| \in \{0, \dots, U\}$.

%% file: content/exact-objective.tex

\section{Exact Local Distributions on Composite Predicates}\label{sec:localdist}
In this section and in \pref{sec:splitting} we will show how to satisfy the constraint $\OBJ(x) = \beta$ exactly, with pseudovariance zero.
Our strategy will be to group predicates together into ``composite'' predicates, and then prove that there is a local $(t-1)$-wise uniform distribution  which is moreover supported on variable assignments for which an exact $\beta$-fraction of the predicates within the composite predicate are satisfied.
We'll then apply \pref{thm:random-graph} to the composite predicates.

We begin with an easier proof for the case when there is a pairwise-uniform distribution over satisfying assignments to our predicate in \pref{sec:pair}, and later in \pref{sec:twise} we handle $t$-wise uniform distributions for larger $t$.
While the pairwise-uniform theorem is less general, the proof is simpler and it already suffices for all of our bisection applications.

\subsection{Pairwise-uniform distributions over $\beta$-satisfying assignments}\label{sec:pair}
Recall our setting: we have a Boolean $k$-ary predicate $P:\{\pm 1\}^k\to \{0,1\}$ and a pairwise-uniform distribution $\nu$ over assignments $x \in \{\pm 1\}^k$ such that $\E[P(x)] = \beta$.
The following theorem states that we can extend $\nu$ into a distribution $\theta$ over assignments to groups of $r$ predicates at a time, $\{\pm 1\}^{r \times k}$ so that exactly $(\beta - \eps)\cdot r$ of the $r$ predicates are satisfied by any assignment $y \sim \theta$.

\begin{theorem}\label{thm:pairwise-ex}
    Let $P:\{\pm 1\}^k \to \{0,1\}$ be a $k$-ary Boolean predicate, and let $\nu$ be a pairwise-uniform distribution over assigments $\{\pm 1\}^k$ with the property that $\nu(x)$ is rational for each $x \in \{\pm\}^k$; that is, there exist a multiset $S \subseteq \{\pm 1\}^k$ such that for each $x \in \{\pm 1\}^k$, $\nu(x) = \Pr_{s \sim S}(s = x)$.
    Suppose also that $\E_{x\sim \nu}[P(x)] = \beta$, and that this is more than the expectation under the uniform distribution, so $\beta > \E_{x\sim \{\pm 1\}^k}[P(x)]$.

    Then for any constant $\eps > 0$, there exists an integer $r = O_{\eps,k}(|S|^3)$ and a rational $\tilde \eps \le \eps$ so that there is a pairwise-uniform distribution $\theta$ over assigments to groups of $r$ predicates, $\{\pm 1\}^{r \times k}$ such that exactly $(\beta - \tilde\eps)r$ of the predicates are satisfied by any assigment $y \sim \theta$.
\end{theorem}

Throughout this section we'll refer to the assignments in the support of $\theta$ as \emph{matrices}, with each row of the $r \times k$ matrix corresponding to the assignment for a single coply of the predicate.

Since we have assumed that the probability of seeing any string in the support of $\nu$ is rational, without loss of generality we can assume that $\nu$ is uniform over some multiset $S \subseteq \{\pm 1\}^k$.
As a first guess at $\theta$, one might try to take $r = c\cdot|S|$ for some positive integer $c$, make $c$ copies of the multiset $S$, and use a random permutation of the elements of this multiset to fill the rows of an $r \times k$ matrix.
But this distribution is not \emph{quite} pairwise uniform.
The issue is that because each individual bit is uniformly distributed, every column of the matrix will always be perfectly balanced between $\pm 1$.
Therefore the expected product of two distinct bits in a given column is
\[
    \frac{1}{\binom{c|S|}{2}}\left(\binom{\tfrac{1}{2}c|S|}{2}(-1)^2 + \binom{\tfrac{1}{2}c|S|}{2}(+1)^2 +\left(\frac{1}{2} c|S|\right)^2(+1)(-1)\right) =  - \frac{1}{c|S|-1} \neq 0.
\]
So, the bits within a particular column have a slight negative correlation.

We'll compensate for this shortcoming as follows: we will randomly choose an element $s$ in the support of $\nu$ to repeat multiple times.
This may in turn alter the number of predicates satisfied out of the $r$ copies of $P$, whereas our express goal was to satisfy the exact same number of predicates under every assignment.
To adjust for this, we'll mix in some rows from the uniform distribution over $\{\pm 1\}^k$, where the number of rows we mix in will depend on whether $P(s) = 1$ or $0$.

\begin{proof}
    Let $S$ be a multiset of strings in $\{\pm 1\}^k$ such that $\Pr_{s\sim S}(s = a) = \nu(a)$.
    We will also require a multiset $T \subseteq \{\pm 1\}^k$ which is a well-chosen mixture of $\nu$ and the uniform distribution; the following claim shows that we can choose such a set.
    Here, we take some care in choosing this combination; the exact choice of parameters will not matter until later.
    \begin{claim}\label{claim:mixture}
	For any constant $\eps > 0$, there is a constant $L = O(1/\eps)$ and a constant $R \ge 1$ so that there are multisets $S',T \subseteq \{\pm 1\}^k$ with the following properties: $S'$ is $RL2^k$ copies of $S$, $T$ has size $|T|= |S'| = LR2^k |S|$, and
	\[
	    \Pr_{x \sim T} [P(x) = 1] = \beta -\eps',
	\]
    where $\eps > \eps' \defeq \frac{1}{L}(\beta - \E_{x \sim \{\pm 1\}^k}[P(x)])$.
    \end{claim}
    \begin{proof}
	Let $s = |S|$, and let $U = \{\pm 1\}^k$.
	Suppose that $\eta 2^k$ of $U$'s assignments satisfy $P$.
	Define $T$ to be the multiset given by $2^k (L-1) R$ copies of $S$ and $s R$ copies of $U$.
	We have that
	\[
	    \Pr_{x\sim T}[ P(x) = 1] = \frac{1}{2^k L R s}\left( \beta s \cdot 2^k (L-1) R + \eta 2^k \cdot s R\right) = \beta - \frac{\beta-\eta}{L}.
	\]
	By choosing $L$ large as a function of $\eps$, we can make this probability as small as we want.
    \end{proof}
    For convenience, let $\ell = 2^k L R|S|$.
    Set the number of rows $r = d\ell$ for an integer $d = O_{\eps}(\ell^2)$ to be specified later.
    We also let $a,b_1,b_0,c_1,c_0$ be integers which will specify the number of rows from $S'$, $T$, and the repeated assignment set; we'll set the integers later, but we will require the property that
    \begin{align}
	d = a + b_1 + c_1 = a + b_0 + c_0. \label{eq:sizecon}
    \end{align}
We generate a sample from $\theta$ in the following fashion:

\medskip
    \begin{compactenum}
	\item Sample $s \sim S$, and fill the first $a\ell$ rows with copies of $s$. Call these the $A$ rows.
	\item Set $i = P(s)$, that is $i = 1$ if $s$ satisfies $P$ and $i=0$ otherwise.
	\item Fill the next $b_i\ell$ rows with $b_i$ copies of each string in $S'$. Call these the $B$ rows.
	\item Fill the last $c_i \ell$ rows with $c_i$ copies of each string in $T$. Call these the $C$ rows.
	\item Randomly permute the rows of the matrix.
    \end{compactenum}
    \medskip
    If $\delta \ell$ assignments in $T$ are satisfying and $\beta \ell$ assignments in $S'$ are satisfying, to ensure that the number of satisfying rows are always the same we enforce the constraint
    \begin{align}
	a\ell + b_1\beta \ell + c_1\delta \ell = b_0\beta\ell + c_0 \delta \ell,\label{eq:numsat}
    \end{align}

Now we handle uniformity.
We will prove that all of the degree-$1$ and degree-$2$ moments of the bits in the matrix are uniform under $\theta$.
First, we argue that the degree-$1$ moments are zero, and that the correlation of any two bits in the same row is zero.
    \begin{claim}
	The bits in a single row of $M$ are pairwise uniform.
    \end{claim}
    \begin{proof}
	We can condition on the row type, $A,B$, or $C$.
	For each type of row, there is a multiset $U$ such that $\E[M_{ij}] = \E_{x\sim U}[x_j] = 0$ by the pairwise uniformity of the uniform distribution over $U$.
	The same argument proves the statement for the product of two bits in a fixed row.
    \end{proof}

    Thus, it suffices to prove that the bits in each column are pairwise-uniform; this is because the pairwise uniformity of rows implies that we can fix the values of any entire column, and the remaining individual bits in other columns will remain uniformly distributed.
    So we turn to proving that the columns are pairwise-uniform.
    \begin{claim}
If we choose $a,b_0,b_1,c_0,c_1$ so that
	\[
		a(\ell-1) - \beta(b_1 + c_1) - (1-\beta)(b_0 + c_0) = 0,
	    \]
	then the bits in a single row of $M$ are pairwise uniform.
    \end{claim}
    \begin{proof}
	We'll prove this by computing the expected product of two distinct bits, $x$ and $y$, which both come from the $i$th column of $M$.
	We will compute the conditional expectation of $xy$ given the group of rows that $x,y$ were sampled from.

	We first notice that conditioned on $x$ coming from one type of row and $y$ coming from another, $x$ and $y$ are independent of each other, and by the uniformity of individual bits in each group, $\E[xy ~|~ x,y \sim \text{ different groups}] = 0$.

    Restricting our attention now to pairs of bits from within the same group, we compute the conditional expectations.
    If both bits come from the $A$ rows, they are perfectly correlated.
    On the other hand, if both bits come from the $B$ or $C$ rows their correlation is as we computed above, $-\frac{1}{b\ell - 1}$ or $-\frac{1}{c\ell - 1}$ respectively.
    Therefore we can simplify
    \begin{align}
	\E[xy]
	& = \beta \cdot \E[xy ~|~ P(a) = 1] + (1-\beta)\cdot\E[xy~|~ P(a) = 0] \nonumber\\
	& =
	+ \beta \left(\frac{\binom{a\ell}{2}}{\binom{d\ell}{2}} ~+~ \E[xy ~|~ x,y \sim B, P(a) = 1] \cdot\frac{\binom{b_1\ell}{2}}{\binom{d\ell}{2}} ~+~
	\E[xy ~|~ x,y \sim C, P(a) = 1] \cdot \frac{\binom{c_1\ell}{2}}{\binom{d\ell}{2}}\right) \nonumber \\
	&\quad + (1-\beta)\left(\frac{\binom{a\ell}{2}}{\binom{d\ell}{2}} ~+~ \E[xy ~|~ x,y \sim B, P(a) = 0] \cdot \frac{\binom{b_0\ell}{2}}{\binom{d\ell}{2}}
	~+~ \E[xy ~|~ x,y \sim C, P(a) = 0] \cdot \frac{\binom{c_0\ell}{2}}{\binom{d\ell}{2}}\right) \nonumber\\
	& = \frac{\ell}{2\binom{d\ell}{2}}\bigg(a(a\ell -1) - \beta (b_1+c_1) -(1-\beta) (b_0 +c_0)\bigg), \label{eq:nocor}
    \end{align}
	where \pref{eq:nocor} gives us the condition of the claim.
    \end{proof}

Finally, we are done given that we can find positive integers satisfying the constraints
\begin{align*}
    d-a &= b_1 + c_1 = b_0 + c_0 & (\text{from \pref{eq:sizecon}})\\
    0 &= a + \beta(b_1 -b_0) + \delta(c_1 - c_0) &(\text{from \pref{eq:numsat}})\\
    0 &= a(\ell -1) - \beta(b_1 + c_1) - (1-\beta)(b_0 + c_0) &(\text{from \pref{eq:nocor}})
\end{align*}
The following can be verified to satisfy the constraints above:
\begin{align*}
    a &:= 2(\beta - \delta)\ell; &
    b_1,c_0 &:=  ((\beta-\delta)(\ell-1) - 1)\ell; &
    b_0,c_1 &:= ((\beta-\delta)(\ell-1) +1)\ell; &
    d &:= 2(\beta-\delta)\ell^2.
\end{align*}
By our choice of $\ell = 2^k |S| L R$ and since $\beta - \delta = \frac{\beta - \E[P(x)]}{L}$, we can choose $R$ large enough so that $(\beta - \delta)(\ell -1) > 1$, and because $\beta\ell$ and $\delta \ell$ are integers, these are all also positive integers, as required.

We compute the number of satisfied assigments as a function of the total, which is
\[
\beta\frac{b_0}{d} + \delta \frac{c_0}{d}
= \beta - \frac{\eps'}{2} + \frac{1}{\ell} \left(\frac{1+\eps'}{2}-\beta\right).
\]
The conclusion thus holds, with $\tilde\eps \defeq \frac{\eps'}{2} - \frac{1}{\ell} \left(\frac{1+\eps'}{2}-\beta\right)$.
\end{proof}

\subsection{$t-1$-wise uniform distributions over $\beta$-satisfying assignments}
\label{sec:twise}

We now prove the generalization of the statement in the previous section to $t$-wise uniform distributions over $\beta$-satisfying assignments.
\begin{theorem}\label{thm:twise-ex}
    Let $P:\{\pm 1\}^k \to \{0,1\}$ be a $k$-ary Boolean predicate, let $t \ge 2$ be an integer, and let $\nu$ be a $(t-1)$-wise uniform distribution over assigments $\{\pm 1\}^k$ so that there exist a multiset $S \subseteq \{\pm 1\}^k$ such that for each $x \in \{\pm 1\}^k$, $\nu(x) = \Pr_{s \sim S}(s = x)$.
    Suppose also that $\E_{x\sim \nu}[P(x)] = \beta > \E_{x\sim \{\pm 1\}^k}[P(x)]$.

    Then for any constant $\eps > 0$, there exists an integer $r = O_{\eps,k}(|S|^4)$ and a rational $\tilde \eps \le \eps$ so that there is a $(t-1)$-wise uniform distribution $\theta$ over assigments to groups of $r$ predicates, $\{\pm 1\}^{r \times k}$ such that exactly $(\beta - \tilde\eps)r$ of the predicates are satisfied by any assigment $y \sim \theta$.
\end{theorem}

The proof will use a similar, though slightly more involved, construction of $\theta$ than in the pairwise case.
It may be helpful to note that the choice of $t=3$ in \pref{thm:twise-ex} will not give the same construction as in \pref{thm:pairwise-ex} (although of course one could set $t = 3$ and obtain a result for pairwise-uniform $\nu$).
In particular, it will not be enough to choose one string to repeat many times in order to improve the column-wise correlations.
Instead, we will repair the correlations in one column at a time, by sampling some subset of the bits in each column from a bespoke distribution, designed to make the columns $(t-1)$-wise independent.
We will have to be careful with the choice of distribution, so that we can still control the number of satisfying assigmnets in $M$ as a whole.

\begin{proof}
    As in the proof of \pref{thm:pairwise-ex}, we will require a well-chosen convex combination of $\nu$ and the uniform distribution to ensure that the number of satisfying assingments is always the same.
    We appeal to \pref{claim:mixture}, taking $S'$ and $T$ to be as described there, with $L = O(1/eps)$ (to be set more precisely later) and $R = 1$.
    For convenience let's let $\ell \defeq 2^k L |S|$ and let's let $\delta = \beta - \eps'$.

    We also call $S'_{i=1}$ and $S'_{i=-1}$ to be the sub-multisets of $S'$ which have the $i$th bit set to $1$ and $-1$ respectively.
    We notice that a uniform sample from $S'_{i=1}$ is equivalent to a uniform sample from $\nu$ conditioned on the $i$th bit being $1$.
    Also by the $(t-1)$-wise uniformity we have $|S'_{i=1}| = \ell/2$.
    Notice that since $S'$ is made up of $2^k L$ copies of $S$, the discrepancy in the number of satisfying assignments between $S_{i=1}$ and $S_{i=-1}$ is always an integer multiple of $2^k L$.

    Set $r$, the number of rows, be an integer which we will specify later.
    We also choose the integer $a$ to represent the size of the correction rows, and $b_n,c_n$, the number of copies of $S'$ and $T$ for each $n \in [ak\ell/(2^kL)]$ (where $n2^k L$ is the number of satisfying assignments in the correction rows).
    To make sure the number of rows always adds up to $r$, we'll need the constraint,
    \begin{align}
	r = ak\ell + b_n\ell + c_n\ell \quad  \forall z \in \{\pm 1\}^k \label{eq:sizecon-t}
    \end{align}

    In order to make sure that the columns are $(t-1)$-wise independent, we require a ``column repair'' distribution $\kappa$ over $\{\pm 1\}^{a\ell}$.
    We will specify this distribution later; for now, we need only that $\kappa$ is symmetric and that the number of $1$s in any $z \sim \kappa$ is a multiple of $\ell/2$.
    The latter property is because, when we choose some part of column $i$ according to $\kappa$, we will want to fix the rows with copies of $S_{i=1}$ and $S_{i=-1}$.

    \medskip
    We generate a sample $M \in \{\pm 1\}^{r\times k}$ from $\theta$ in the following fashion:
    \medskip
    \begin{compactenum}
    \item For each $i \in [k]$, independently sample a string $z_i \sim \kappa$.
	Add $a\ell$ rows to $M$, where in the $i$th column we put the bits of $z_i$, and we set the remaining row bits so that if $z_i$ has $(a-a')\ell/2$ entries of value $1$ and $a'\ell/2$ entries of value $-1$, then we end up with $a'$ copies of $S'_{i=-1}$ and $a-a'$ copies of $S'_{i=1}$. Call these rows $A_i$.
    \item Compute the integer $n$ such that $n\cdot 2^kL$ is the number of rows in $\cup_{i=1}^k A_i$ containing satisfying assignments to $P$, given our choices of $z_i\ \forall i \in [k]$.
	\item Add $b_n\ell$ rows to $M$ which contain $b_n$ copies of each string from $S$. Call these rows $B$.
	\item Add $c_n\ell$ rows to $M$ which contain $c_n$ copies of each string from $T$. Call these rows $C$.
	\item Randomly permute the rows of $M$.
    \end{compactenum}
    \medskip
    So that the number of satisfying assignments $\Lambda$ is always the same, we require that
    \begin{align}
	\Lambda
	 = n2^kL + b_n \cdot \beta \ell + c_n \cdot \delta \ell \quad \forall n \in [ak\ell/(2^kL)]
	 \label{eq:objcon-t}
    \end{align}

    Now, we will derive the conditions under which $(t-1)$-wise independence holds.
    As above, we first consider bits that are all contained in a fixed row.
    \begin{claim}
	The bits in a single row of $M$ are $(t-1)$-wise uniform.
    \end{claim}
    \begin{proof}
	We can condition on the row type, $A_1,\ldots,A_k, B$, or $C$.
	Sampling a uniform row from $B$ or $C$ is equivalent to sampling from $\nu$ or $\gamma$, which are $(t-1)$-wise uniform.
	Since $\kappa$ is symmetric, sampling a row from $A_i$ is equivalent to sampling from $\nu$ as well, and we are done.
    \end{proof}

    From the claim above, if we condition on the value in $d < t-2$ columns, the remaining $t-2-d$ columns will remain identically distributed; this is because the rows are $(t-1)$-wise uniform, so after conditioning the distribution in each row will remain $(t-1-d)$-wise uniform.
    Thus, proving that each column is $(t-1)$-wise uniform suffices to prove $(t-1)$-wise uniformity on the whole.

    The following lemma states that we may in fact choose $\kappa$ so that this condition holds exactly.
    \begin{lemma}\label{lem:kappa}
	Let $y \in \{\pm\}^{r - a\ell}$ be a perfectly balanced string.
	If $a\ell > h_1 \cdot \sqrt{tr}$ for a fixed constant $h_1$ and $\sqrt{r} \ge (t-1)\ell 2^{h_2 t}$ for a fixed constant $h_2$, then there is a distribution $\kappa$ over $\{\pm 1\}^{a\ell}$, supported on strings which have a number of $1$s which is a multiple of $\ell/2$, such that if $x$ is sampled by choosing $z \sim \kappa$, concatenating $z$ with $y$ and applying a random permutation, then for any $S \subset [r]$ with $|S| \le t-1$, $\E[x^S] = 0$.
    \end{lemma}
    Since each column is distributed as the string $x$ described in the lemma statement, the lemma suffices to give us $(t-1)$-wise uniformity of the columns.
    We'll prove the lemma below, but first we conclude the proof of the theorem statement.

    We now choose the parameters to satisfy our constraints.
    We have the requirements:
    \begin{align*}
	\Lambda &= n2^kL + b_n \beta \ell + c_n \delta \ell \quad \forall n \in [ak\ell/(2^kL)] & \text{(from \pref{eq:objcon-t})}\\
	r - ak\ell &= b_n\ell + c_n\ell \quad \forall n \in [ak\ell/(2^kL)] & \text{(from \pref{eq:sizecon-t})}\\
	a\ell &\ge h_1 \sqrt{tr} & \text{(from \pref{lem:kappa})}\\
	\sqrt{r} &\ge (t-1)\ell 2^{h_2 t} & \text{(from \pref{lem:kappa})}
    \end{align*}
    where $h_1$ and $h_2$ are universal constants.
    The below choice of integer parameters satisfies these requirements, as well as the requirement of always being non-negative:
    \begin{align*}
	u &= \left(\beta - \E_{x \sim \{\pm 1\}^k}[P(x)]\right)2^k |S|; &
	L &= u\cdot \max(1, \lceil h_1 + h_2\rceil )\cdot \left\lceil\frac{1}{\eps}\right\rceil\cdot k; &
	\ell &= 2^k L|S|;\\
	a &= \left\lceil h_1 2^{h_2 t/2}t k \right\rceil; &
	r &= \left\lfloor\frac{a^2}{h_1^2 t}\right\rfloor \cdot \ell^2;&
	\\
	b_0 &= \frac{1}{2}\left(\frac{1}{\ell}r-ak\right);&
	b_{n+1}&= b_n - \frac{2^{k}L}{u};\\
	c_0 &= \frac{1}{2}\left(\frac{1}{\ell}r-ak\right);&
	c_{n+1} &= c_n + \frac{2^{k}L}{u}.
    \end{align*}
    Finally, we have that the fraction of satisfying rows in $M$ is always exactly
    \[
	\frac{\Lambda}{r}
	= \frac{\frac{1}{2}\left(r-ak\ell\right)\beta + \frac{1}{2}\left(r-ak\ell\right)\delta}{r}
	= \beta - \frac{\eps'}{2} - O\left(\frac{ak\ell}{r}\right).
    \]
The latter term is $O(\frac{1}{\ell})$, and we have chosen $L$ large enough so that it is smaller than $\eps'/2$.
\end{proof}

\begin{proof}[Proof of \pref{lem:kappa}]
    For convenience, call $m \defeq r - a\ell$.
    Recall that we take $x$ to be sampled by taking a balanced string $y \in \{\pm 1\}^{m}$, sampling $z \sim \kappa$, appending $z$ to $y$ and then applying a uniform permutation to the coordinates.

    We will solve for $\kappa$ with a linear program (LP) over the probability $p_z$ of each string $z \in \{\pm 1\}^{a\ell}$.
We have the program
    \begin{align*}
	\forall S \in [m], |S| \in \{1,\ldots,t-1\} :&& \sum_{\substack{z \in \{\pm 1\}^{a\ell} \\ (\ell/2) | \sum_j z_j }}\E\left[x^S ~\bigg{|}~ \sum_i x_i = \sum_{j \in [a\ell]} z_j\right] \cdot p_z &= 0\\
	\forall z \in \{\pm 1\}^{a\ell} \ s.t.\ \frac{\ell}{2} \bigg{|} \sum_j z_j :&& p_z &\ge 0
    \end{align*}
    Since we can take any solution to this LP and scale the $p_z$ so that they sum to $1$, the feasibility of this program implies our conclusion.
    So suppose by way of contradiction that this LP is infeasible.
    Then Farkas' lemma implies that there exists a $q \in \R^{t-1}$ such that
    \[
	\forall z \in \{\pm 1\}^{a\ell} \ s.t.\ \frac{\ell}{2} \bigg{|} \sum_j z_j,\quad \sum_{\substack{S \subseteq [m]\\|S| \in \{1,\ldots,t-1\}}} \E\left[x^S ~\bigg{|}~ \sum_i x_i  = \sum_{j \in [a]}z_j\right] \cdot y_s > 0.
    \]
Without loss of generality, we scale $q$ so that $\sum_S q_S^2 = 1$.
    Moreover by the symmetry of the expectation over subsets $S$, we can assume that $q_{S} = q_{T}$ whenever $|S| = |T|$.
    This implies that the degree-$t$ mean-zero polynomial
    \[
	q(x) = \sum_{\substack{S\subseteq [m]\\ 1\le |S| \le t-1}} q_{S} \cdot \chi_S(x)
	\]
	has positive expectation over every layer of the hypercube with $\left|\sum_i x_i\right| = d$ such that $d \le a\ell$ and $\frac{\ell}{2} | d$.
	Furthermore, $q$ is a symmetric polynomial, which implies that it takes the same value on all inputs of a fixed Hamming weight; this implies that it takes positive values on every inputs $x$ with $\left|\sum x_i \right| \in [2a]\cdot \frac{\ell}{2}$.

    The following fact will give us the contradiction we desire:
\begin{fact}[Tails of low-degree polynomials \cite{DFKO07}, see Theorem 4.1 in \cite{MR2765709-Austrin11}]\label{fact:fourier-tails}
Let $f:\zo^m \rightarrow \R$ be a degree-$t$ polynomial with mean zero and variance $1$.
Then, there exist universal constants $c_1, c_2 > 0$  such that $\Pr[ p \leq -2^{-c_1t}] \geq 2^{-c_2t}.$
\end{fact}

    We will show that since $q$ takes positive value on every hypercube slice of discrepancy $\frac{\ell}{2} \cdot [2a]$, this implies that it takes positive values on most hypercube slices with discrepancy at most $a\ell$.
    Because we have chosen $a \ell$ so that this comprises the bulk of the hypercube, this in turn will contradict \pref{fact:fourier-tails}.

    In fact, because $q$ is symmetric and of degree $t-1$ over the hypercube, we can equivalently write $q$ as a degree-$(t-1)$ polynomial in the single variable $x' := \sum_i x_i$, $q(x) = g(\sum_i x_i) = \sum_{s \in \{0,\ldots, t-1\}} g_s \cdot \left(\sum_i x_i\right)^s$.
    Viewing $g$ as a univariate polynomial over the reals, $g$ has at most $t-1$ roots.
    Therefore we conclude that $q$ can only be non-positive on at most $t-1$ intervals of layers of discrepancy $(i\ell/2, (i+1)\ell/2)$.
    So for at least $2a\ell - (t-1)\frac{\ell}{2}$ slices of the hypercube around $0$, $q$ takes positve value.

    Each slice of the hypercube has probability mass at most $\frac{1}{\sqrt{\frac{\pi}{2}(m + a\ell)}}$.
    By our choice of $a,\ell, m$ and by a Chernoff bound,
    \[
	\Pr(q(x) > 0) \ge \Pr\left(\left|\sum_i x_i\right| \le a\ell\right) - \frac{(t-1)\ell}{2\sqrt{m + a\ell}} > 1 - 2^{-c_2 t},
    \]
which contradicts \pref{fact:fourier-tails}.
\end{proof}

%% file: content/splitting.tex
\section{SOS lower bounds for CSPs with exact objective constraints}\label{sec:splitting}

In this section we put things together and show how to extend \pref{thm:twise-ex} to prove \pref{thm:main2}.
As discussed briefly in \pref{sec:mainthms}, our random instance of the $CSP(P^{\pm})$ will be sampled in a somewhat non-standard way, which we will refer to as ``batch-sampling.''
This is because, in order to apply \pref{thm:twise-ex} to a random instance $\Phi$ of a Boolean CSP, we need to partition $\Phi$'s constraints into groups of $r$ non-intersecting constraints for some integer $r$, while also maintaining the expansion properties required by \pref{thm:good-if-expander}.

We first prove \pref{thm:main2} as stated, for random CSPs sampled from a slightly different distribution.
Then in \pref{sec:sparse-ex} we show that for a ``standard'' random CSP with $m = o(n^{3/2})$ constraints, we can still get a theorem along the lines of \pref{thm:main2}.

\subsection{Exact objective constraints for batch-sampled random CSPs}

Suppose that $P$ is a $k$-ary predicate, and let $r$ be some positive integer which divides $m$.
We'll ``batch-sample'' an $n$-variate random $CSP(P^{\pm})$ with $m$ clauses as follows:
\begin{enumerate}
    \item Choose independently $m/r$ subsets each of $r\cdot k$ distinct variables uniformly at random from $[n]$, $S_1,\ldots,S_m$
    \item For each $j \in [m/r]$, $S_j = \{x{i_1},\ldots,x_{i_{rk}}\}$:
	\begin{itemize}
	    \item Choose a random signing of $P$, $z_j \in \{\pm 1\}^k$
	\item To each block of $k$ variables in $S_j$, $( x_{i_{(\ell-1)\cdot k + 1}},\ldots,x_{i_{\ell k}})$ for $\ell \in [r]$, add the predicate $P$ with signing $z_j$.
	\end{itemize}
\end{enumerate}

\begin{theorem}[Restatement of \pref{thm:main2}]
    Let $P$ be a $k$-ary predicate, and let $\nu$ be a $(t-1)$-wise uniform distribution over $\{\pm 1\}^k$ under which $\E_\nu[P] = \beta$.
    Then for each constant $\eps > 0$ there is a choice of positive integer $r$ such that for a random instance of $CSP(P^{\pm})$ on $n$ variables with $m = \Delta n$ constraints for sufficiently large $\Delta$ and $r | m$, sampled as detailed above, there is a degree-$\Omega(\frac{n}{\Delta^{2/(t-2)}\log \Delta})$ SOS pseudodistribution whch satisfies with pseudovariance zero the constraint $\OBJ(x) = \beta - \eps_r$, where $\eps_r < \eps$.
    This is also true when cardinality constraints are imposed as in \pref{thm:main1}.
\end{theorem}
\begin{proof}
    This distribution over instances is equivalent to the standard notion of sampling a random CSP with $m/r$ constraints in the composite predicates from \pref{thm:twise-ex}: a scope is chosen independently and uniformly at random for each predicate.
    Therefore, if we replace each collection of constraints corresponding to $S_j$ with the composite predicate from \pref{thm:twise-ex}, and modify $\nu$ in accordance with the signing $z_j$, we have a $(t-1)$-wise uniform distribution over solutions to the composite predicates supported entirely on assignments which satisfy exactly $\beta - \eps_r$ of the clauses.
    Combining this with the expansion theorem (\pref{thm:random-graph}), we have our conclusion.
\end{proof}

\subsection{Exact objective constraints for sparse random CSPs.}\label{sec:sparse-ex}

Though the batch-sampled distribution over CSPs for which \pref{thm:main2} holds is slightly non-standard, here we show that with minimal effort, we can prove a similar theorem for sparse random instances sampled in the usual manner, when $m = o(n^{3/2})$.

\begin{theorem}\label{thm:sparse}
    Let $P$ be a $k$-ary predicate, let $\nu$ be a $(t-1)$-wise uniform distribution over $\{\pm 1\}^k$ such that $\E_{\nu}[P] = \beta$, and suppose we sample a random instance $\Phi$ of a $CSP(P^{\pm})$ in the usual way, by selecting $m$ random signed $P$-constraints on $n$ variables.
    Then if $m = \Delta n = o(n^{3/2})$ for sufficiently large $\Delta$, with high probability over the choice of $\Phi$, for each $\eps > 0$ there exists some constant $\eps_{\Phi} \le \eps$ such that there is a degree-$\Omega_{\eps}(\frac{n}{\Delta^{2/(t-2)}\log \Delta})$ pseudodistribution which satisfies with pseudovariance zero the constraint $\OBJ(x) = \beta - \eps_{\Phi}$ and a Hamming weight constraint $\sum_{i\in[n]} x_i = B$ for $|B| = O(\sqrt{n})$.
\end{theorem}

\begin{proof}
    Fix $\epsilon$, and let $r$ be the corresponding constant required to achieve objective $\OBJ(x) = \beta - \eps^*$ under \pref{thm:twise-ex} for $\eps^* < \eps/2$.

We first couple the standard sampling procedure for a random $P$-CSP to a different sampling procedure.
    For simplicity we at first ignore the possible signings of $P$, and assume we work only with un-negated variables; later we explain how to modify the proof to accomodate negative literals.

    We sample a random CSP by independently and uniformly choosing $m$ random scopes $S_1,\ldots,S_m$.
    For each $\ell \in \{0,1,\ldots \lfloor m/r\rfloor -1\}$, the probability that $S_{\ell r + 1},\ldots,S_{(\ell + 1)r}$ have non-intersecting scopes is at least
    \[
	\Pr[\cap_{j\in[r]} S_j = \emptyset]
	= \prod_{i=2}^r \Pr[S_i \cap \left(\cap_{j < i} S_j\right) = \emptyset ~|~\cap_{j < i} S_j = \emptyset]
	= \prod_{i=2}^r \left(1-i\frac{k}{n}\right) \ge 1 - O\left(\frac{r\cdot k}{n}\right).
    \]
So with high probability for all but $O(\frac{m}{n})$ of the intervals of constraints $j \in [\ell\cdot r + 1, (\ell + 1)r]$, the constraints will be non-intersecting.
Call this the ``non-intersecting configuration''.

    Define a ``collision configuration'' to be a choice of scopes for which the above condition does not hold; that is, a specific way in which $S_j$ intersects with one or more $S_{j'}$ when $j,j' \in  [\ell\cdot r + 1, (\ell + 1)r]$.
    Each of the $\approx \binom{2kr}{r}$ collision configurations has a fixed probability of ocurring (which may be easily calculated), and the total sum of these probabilities is at most $O(r\cdot k/n)$.

    Let $\cD_r^{(m)}$ be the multinomial distribution which describes the number of occurrences of each configuration for a random CSP with $m$ constraints ($\lfloor m/r \rfloor$ configurations).
    We couple the standard sampling procedure with the following alternative sampling procedure: we first sample $c \sim \cD_r^{(m)}$ to determine how many configurations of each type there are.
    Then, for each collision configuration specified by $c$, sample the scope (of size $< k\cdot r$) for each of the collision configurations independently and uniformly at random.
    Also, sample and additional $(m \mod r)$ scopes of $k$ variables for the ``leftover copies'' of $P$.
    Finally, sample the scopes of the non-intersecting configurations specified by $c$ independently uniformly at random.
    The coupling of the two processes is immediate.

    Let $C$ be the number of collision configurations plus $(m\mod r)$, the number of leftover copies.
    As shown above, with high probability over $c \sim \cS_{r}^{(m)}$, the number of collision configurations is at most $O(m/n) = o(n^{1/2})$, so $C = o(n^{1/2}) = o(m)$.

    From our alternate sampling procedure, we conclude that with high probability we can meet the conditions of \pref{thm:random-graph} by fixing an arbitrary variable assignment to any collision configuration.
    That is, we could alternately first sample the collision configurations and leftover copies, and then set all of the variables present inside be set to (say) False.
    We take note of how many constraints in the $P$-CSP are and are not satisfied by this unary assignment, and we correspondingly amend $\eps^*$ to $\eps_{\Phi}$.
    Since with high probability at most $o(m)$ constraints are fixed, we retain the property that $\eps_{\Phi} \le \eps^* + o(1) \le \eps$.

    Now, if we wish to satisfy a Hamming weight constraint, we add arbitrary matching and unary constraints to get the desired Hamming weight; at most $O(C)$ unary constraints are needed to compensate for the $\le C kr$ variables we set to False.

    Finally, we sample the remaining non-intersecting configurations independently; by \pref{thm:random-graph} when $C = o(n^{1/2})$, the expansion properties we require are met for the composite predicates on the non-intersecting configurations.
    Since this occurs with high probability, we are done.

    To extend the argument to allow predicates on negative literals, we couple with a slightly more elaborate sampling procedure: for each signing pattern $z \in \{\pm 1\}^k$, we draw a separate set of $m_z$ predicates (where $m_z$ may either be deterministic or sampled from a multinomial distribution).
    For each signing separately we repeat the argument above, and then in the final sampling procedure we sample counts $c_z$ for each signing $z$, add the leftover copies and collision configurations separately for each signing, add the unary constraints, and then sample the remaining non-intersecting copies.
\end{proof}

%% file: content/min-bi-feige.tex
\section{$\Omega(\sqrt n)$ round SOS lower bound for Min-Bisection}\label{sec:feige}
In this section, we apply \pref{thm:main2} to prove the first part of \pref{thm:apps}, which replicates Feige's reduction from random 3AND to Min-Bisection \cite{Fei02} in the context of SOS.
Specifically, we will prove the following result:
\begin{theorem}
\label{thm:feige-min-bisection}
For any constant $\epsilon$, there is an easy-to-sample distribution on $n$-vertex Min-Bisection instances such that, with high probability over the choice of instance, the following hold:
\begin{itemize}
\item There is a degree-$\Omega(\sqrt n)$ SOS solution with objective value $\frac{3}{4} - \eps$, where the bisection identity is satisfied exactly with pseudovariance zero.
\item Every genuine solution has objective value at least $1 - \epsilon$.
\end{itemize}
In particular, degree-$o(\sqrt n)$ SOS does not deliver a factor-($\tfrac{4}{3} - \epsilon$) approximation for Min-Bisection.
\end{theorem}

We note that there is also an alternative proof of \pref{thm:feige-min-bisection} which does not require \pref{thm:main2}, but also does not follow Feige's reduction as faithfully---see \pref{rem:xor} at the end of the section.

\begin{proof}
Our reduction almost exactly follows that of Feige \cite{Fei02}.
    We begin by sampling a random 3AND instance $\Phi$: for even $N = \Theta(\sqrt{n})$, we take $N$ Boolean variables $x \in \{\pm 1\}^N$, and impose the constraint $\sum_i x_i = 0$.\footnote{This is in place of using literals, which gives an equivalent proof but is a bit more notationally involved.}
    Then for a sufficiently large constant $\Delta$, we sample $M = \Delta N$ random un-negated 3AND clauses on the variables.\footnote{So we only add clauses of the form $(x_i \wedge x_j \wedge x_k)$, with no variable negations---again this is done to ease notation and is enabled by the $\sum x_i = 0$ constraint.}

    The uniform distribution over odd-parity strings is pairwise-uniform and satisfies 3AND with probability $\frac{1}{4}$.
    Applying \pref{thm:sparse} (the sparse version of \pref{thm:main2}), we conclude that there is some $\eps_{\Phi} \le \frac{1}{2}\epsilon$ (which is easily computable) such that the degree-$\Omega(n)$ SOS relaxation satisfies with pseudovariance zero the equality $\OBJ(x) = \frac{1}{4} - \eps_{\Phi}$ and the equality $\sum_i x_i = 0$.
    \medskip

    Now, we construct a Min-Bisection instance $G = (V,E)$ on $n$ vertices as follows: for each variable $x_i$, we create a vertex $v_i$.
    For each clause $C_j$, we create a clique on $M' = 3M + 1$ vertices $y_{j}^{(1)},y_{j}^{(2)},\ldots, y_{j}^{(M')}$.
    We also designate the vertex $y_j \defeq y_{j}^{(1)}$ as a special vertex, and add an edge from $y_{j}$ to the variable vertex $v_i$ for each $x_i$ which is present in $C_j$.\footnote{for instance if $C_j = (x_a \wedge  x_b \wedge x_c)$, we add the edges $(y_{j}, v_a), (y_{j}, v_b), (y_{j}, v_c)$.}
    Finally, we add a ``giant'' clique on $M'' = M \cdot M' \cdot \frac{(1+\eps_{\Phi})}{2}$ extraneous vertices $z_1,\ldots,z_{M''}$.

\paragraph{Soundness}
 We replicate Feige's soundness analysis~\cite{Fei02} for the reader's convenience.
    It is easy to check that the bipartition which places half the variables and $\frac{3 + \eps_{\Phi}}{4}M$ clause cliques in one partition cuts no more than $3M$ edges.
For this reason, the true minimum bisection can never separate any two vertices in a clause gadget, since that already cuts $M' > 3M$ edges.
Similarly, no two vertices in the giant clique can be separated.
    It follows that the true min bisection must exactly bisect the variable vertices, and must split the clause gadgets into one group of $\frac{1-\eps_{\Phi}}{4}M$ cliques and one group of $\frac{3+\eps_{\Phi}}{4}M$ cliques.

Now if $\Delta$ is chosen as a sufficiently large function of $\epsilon$, a Chernoff + union bound implies that every subset of $N/2$ variable vertices is adjacent to at least $\frac{3M}{2}(1-\eps)$ clause vertices with high probability.
    By a similar Chernoff + union bound, for any bipartition of the variable vertices into two sets $S,T$ each of size $N/2$, there will be $\frac{1}{8}(1\pm \eps)M$ clause vertices which have $3$ (and by symmetry $0$) edges in $S$, and $\frac{3}{8}(1\pm \eps) M$ that have two (and by symmetry $1$) neighbors in $S$.
Counting the number of edges cut in this bipartition from the variable vertex side: we take the sum of degrees of the variable vertices in $S$, then choose the best choice of $(1-\eps)/4M$ clauses to include in $S$'s side of the bipartition, for which we get to subtract the edges to the $\approx 1/8$ clauses whose neighbors are all included into $S$, then subtract $2$  and add 1 for the remaining $\approx 1/8$ clauses that have two neighbors in $S$ and one in $T$.
So with high probability the minimum bisection will cut at least $\frac{3(1-\eps)}{2}M - 3\frac{1+\eps}{8}M - (2-1) \frac{1+\eps}{8} = (1-2\eps)M$ edges.
\paragraph{Completeness}
    There is a degree-3 reduction between the variables $x_1,\ldots,x_N$ in the SOS program for $\Phi$ and the variables in the SOS program of the Min Bisection instance, $U = \{v_i\}_{i\in[N]}\cup \{y_j^{(a)}\}_{j\in[M],a\in[M']} \cup \{z_k\}_{k\in[M'']}$.
    So losing a factor of $3$ in the SOS degree the pseudoexpectation for $\Phi$ translates to a pseudoexpectation for $G$ as follows:
    \begin{itemize}
	\item For each $i \in [N]$, set $v_{i} = x_i$.
	\item For each $j\in[M], a \in [M']$, set $y_{j}^{(a)} = 2\Ind(C_j(x) =1) - 1$ (the $\pm 1$ indicator that $C_j$ is satisfied).
	\item Finally, set $z_a = 1$ for all $a \in [M'']$.
    \end{itemize}
    For integral $x_i$ with objective value $\frac{1}{4} - \eps_{\Phi}$, this corresponts to putting the True vertices, satisfied clause gadgets, and the giant clique on one side, and the False vertices and unsatisfied clause gadgets on the other.
    Since we imposed the ideal constraints $x_i^2 = 1$ in the pseudoexpectation for $\Phi$, it is not hard to check that the ideal constraint $u^2 = 1$ also holds for each $u \in U$.

    We now check that the bisection constraint $\sum_{u\in U} \pE[u] = 0$ is satisfied with pseudovariance $0$.
    For any polynomial $Q(u)$ of degree at most $O(N)$, we have that
\begin{align*}
    \pE\left[\left(\sum_{u_i \in U} u_i\right) Q(u)\right]&=
    \parens*{\sum_{a \in [M'']} \pE[z_a\cdot Q(u)]} + \parens*{\sum_{a \in [M']}\pE\left[\left( \sum_{j \in [M]} y_{j}^{(a)}\right)\cdot Q(u)\right]} + \parens*{\pE\left[\left(\sum_{i \in [N]} v_i\right)\cdot Q(u)\right]}\\
    &= \left(M''\cdot \pE[Q(u)]\right) + M'\cdot \left( - \frac{1 + \eps_{\Phi}}{2} \cdot M \cdot \pE[Q(u)]\right) + \left(0\cdot \pE[Q(u)]\right),
\end{align*}
    and since we have set $M'' = \frac{1+\eps_{\Phi}}{2} M \cdot M'$, the above is exactly 0, as desired.

Finally, we check the objective value.
Our objective function is,
    \[
	\OBJ(u) = \frac{1}{4} \sum_{(u,v) \in E} ( u - v)^2.
    \]
    In our pseudoexpectation, we have $y_{j}^{(a)} = y_{j}^{(b)}$ for all $a,b \in [M']$, and we have $z_{a} = z_{b}$ for all $a,b \in [M'']$, so these edges do not contribute to our objective value.
   Applying this observation,
    \begin{align}
	\pE[\OBJ(u)]
	&= \frac{1}{4} \sum_{j \in [M]} \sum_{a : x_a \in C_j} \pE[(y_j - v_a)^2]
	= \frac{1}{4} \sum_{j \in [M]} \sum_{a: x_a \in C_j} (2 - 2\pE[y_jv_a]) \label{eq:obj}
    \end{align}
    where we have used the $u^2 =1$ ideal constraints.
    Now, we have $y_j = 2\Ind(C_j(x) = 1) - 1$.
    Since $\Ind(C_j(x) = 1)$ is zero when $v_a = -1$ and one when $v_a = 1$, mod the ideal $v_a^2 = 1$, we have that $\Ind(C_j(x) = 1) \cdot v_a = \Ind(C_j(x)=1)$.
    Therefore,
	\begin{align}
	    v_a y_j
	    &= 2\Ind(C_j(x) = 1) - v_a.
\label{eq:simp}
	\end{align}

    Applying \pref{eq:simp} to \pref{eq:obj},
    \begin{align*}
	\pE[\OBJ(u)]
	&= \frac{1}{4} \sum_{j \in [M]} \sum_{a : x_a \in C_j} (2 - 4\pE[\Ind(C_j(x)) = 1]  + 2\pE[v_{a}])\\
	&= \frac{3}{2}M - 3 \cdot \sum_{j \in [M]} \pE[\Ind(C_j(x)) = 1] + \frac{1}{2} \sum_{a \in [N]} \deg_\Phi(a) \cdot \pE[v_{a}]\\
	&= \frac{3}{2}M - \frac{3(1-\eps_{\Phi})}{4}M + \sum_{a \in [N]} \deg_\Phi(a)\cdot \pE[x_{a}],
    \end{align*}
    where we have used $\deg_{\Phi}(a)$ to denote the degree of $x_a$ in the clause-variable incidence graph for $\Phi$.
    To obtain the second line we have re-indexed the summation, and in the final line we have used our objective constraint for $\Phi$.

    To finish the proof, it remains only to bound the final sum.
    We have that
    \[
    \sum_{a \in [N]} \deg_\Phi(a)\cdot \pE[x_{a}]
    = \frac{3M}{N} \sum_a \pE[x_a] + \sum_{a \in [N]} \left(\deg_\Phi(a) - \frac{3M}{N}\right) \cdot \pE[x_a].
    \]
The first term is $0$ by the balancedness constraint.
By the constraint that $\pE[x_a^2] = 1$ and by the Cauchy-Schwarz inequality, $|\pE[x_a]|\le 1$.
By a Chernoff + union bound, for $\Delta$ chosen as a sufficiently large function of $\eps$, the quantity $\sum_a \left|\deg_\Phi(a) -\frac{3M}{N} \right|$ cannot be more than $\eps M/2$ with high probability.
The conclusion follows.
\end{proof}

\begin{remark}\label{rem:xor}
    We observe that one could prove the same theorem in an alternate way which does not make use of \pref{thm:main2} (but does not follow Feige's reduction).
    This is done by exploiting the fact that degree-$\Omega(n)$ SOS assigns value exactly $1$ to random 3XOR instances \cite{Gri01,Sch08}.

    Instead of starting with a random 3AND instance, one would instead sample a random 3XOR instance, and then replace each 3XOR predicate $P(x_{i_1},x_{i_2},x_{i_3})$ with four 3AND predicates which correspond to the satisfying assignments of $P$.
    Since a degree-$\Omega(N)$ SOS solution has value $1$ for 3XOR, the sum of values of the four 3AND predicates would be exactly $\frac{1}{4}$.
    Notice however that this 3AND instance is highly non-random, and the soundness analysis becomes slightly more involved.
\end{remark}

%% file: content/bisection.tex
\section{Linear round SOS gaps for Max- and Min-Bisection}\label{sec:bisect}
In this section we will  show that there are linear-round SOS integrality gaps for Max- and Min-Bisection, proving the second part of \pref{thm:apps}.
We'll start Max-Bisection.
\begin{theorem}                                     \label{thm:max-bisection-main}
    For any constant $\eps$ of the form $\frac{1}{4c}$ for $c > 0$ an even integer, there is an easy-to-sample distribution on $n$-vertex Max-Bisection instances such that, with high probability over the choice of instance, the following hold:
    \begin{itemize}
        \item There is a degree-$\Omega(n)$ SOS solution with objective value
         $\frac{3}{4} - \eps$, where the bisection identity and objective value identity are satisfied exactly with pseudovariance zero.
        \item Every genuine solution has objective value at most $\frac{11}{16} + \eps$.
    \end{itemize}
    In particular, degree-$o(n)$ SOS does not deliver a factor-$(\frac{11}{12}+\eps)$-approximation for Max-Bisection.
\end{theorem}

The proof involves a construction (originating in~\cite{TSSW00}; see also~\cite{OW12}) that one might use for factor-$\frac{11}{12}$ 2XOR hardness.  We will think of Max-Bisection not in graph terms, but rather as a binary CSP with ``$\neq$'' constraint and the requirement that assignments be exactly balanced.
The result for Min-Bisection will follow from a near-identical proof, which we save for \pref{sec:min-bi}
\subsection{Soundness results}  \label{sec:soundness}
Let $P : \{\pm 1\}^k \to [0,1]$ be a generalized Boolean predicate (``generalized'' because its range is $[0,1]$, not $\{0,1\}$).  The example to keep in mind is $k = 4$,
\begin{equation}    \label{eqn:basic-P}
    P(x_1, x_2, x_3, x_4) = \begin{cases}
                                            \frac12 & \text{if $\absryan{\littlesum_i x_i} = 0$,} \\
                                            \frac34 & \text{if $\absryan{\littlesum_i x_i} = 2$,} \\
                                            1 & \text{if $\absryan{\littlesum_i x_i} = 4$.}
                                       \end{cases}
\end{equation}
In other words, $P(x_1, x_2, x_3, x_4)$ is the maximum fraction of constraints that can be simultaneously satisfied in the following $5$-variable Max-Cut instance: $\{y \neq x_1, y \neq x_2, y \neq x_3, y \neq x_4\}$.

Let us further generalize $P$ by identifying it with its multilinear extension, and allowing it to accept inputs in $[-1,+1]$, rather than just $\{\pm 1\}$.  The interpretation of $P(\mu)$ for $\mu \in [-1,+1]^k$ is the expected value of $P(\bx)$ when $\bx \in \{\pm 1\}^k$ is chosen by letting $\bx_1, \dots, \bx_k$ be independent $\pm 1$ values with $\E[\bx_i] = \mu_i$.  In the above example,
\[
    P(\mu_1, \mu_2, \mu_3, \mu_4) = {\tfrac {11}{16}} +\tfrac{1}{16}\parens*{\mu_{{1}}\mu_{{2}}+\mu_{{1}}\mu_
{{3}}+\mu_{{1}}\mu_{{4}}+\mu_{{2}}\mu_{{3}}+\mu_{{2}}\mu_{{4}}+\mu_{{3
}}\mu_{{4}}} -\tfrac{1}{16} \mu_{{1}}\mu_{{2}}\mu_{{3}}\mu_{{4}}.
\]

Fix a finite subset $D \subseteq [-1,+1]$.  (The example to keep in mind is when $D$ is all integer multiples of some small constant $\eps > 0$.)  The following lemma is a standard analysis of uniformly random CSPs:
\begin{lemma}                                       \label{lem:random-big-csp}
    Consider a randomly chosen ``CSP'', with a set $\calX$ of $n$ variables to be assigned values in~$D$, and $m = \Delta n$ randomly chosen ``scopes'', where each scope is tuple of $k$~distinct variables.  Given a global assignment $f : \calX \to D$, let $\phi_f$ denote its associated ``empirical probability distribution on~$D$''; i.e., $\phi_f(u)$ is the fraction of $x \in \calX$ such that $f(x) = u$.  Also, let $\Psi_f$ denote the ``empirical probability distribution on scopes''; i.e., the distribution on $D^k$ given by choosing a random scope~$C$ and forming~$f(C)$.

    Then the following statement holds with high probability:  For all global assignments $f : \calX \to D$, the distribution $\Psi_f$ is $\eta_1$-close in total variation distance to the product distribution $\phi_f^{\otimes k}$, where $\eta_1 = O\parens*{\sqrt{\frac{\log |D|}{\Delta}} + \frac{k^2}{n}}$.
\end{lemma}
\begin{proof}
    Fix a global assignment~$f$ and a tuple $U \in D^k$.  When a scope~$\bC$ is randomly chosen, the probability of $f(\bC) = U$ is very nearly equal to~$\phi_f^{\otimes k}(U)$. The difference only comes from the fact that the variables in~$\bC$ are chosen \emph{without} replacement; this produces a total variation distance of $O(k^2/n)$.  Now by a Chernoff bound, we conclude that after all scopes are chosen, $\Psi_f(U)$ will be within $\pm O(k^2/n) \pm O(\sqrt{\log(1/\delta)/m})$ except with probability at most~$\delta$.  If we take, say $\delta = |D|^{-2n}$, this is sufficient to union-bound over all $D^n$ possible~$f$ and all $D^k$ possible~$U$.
\end{proof}

In the following corollary, one should keep in mind the setting~$\zeta = 0$.  We will use the notation $\E[P] = P(0, \dots, 0) = \wh{P}(\emptyset)$.
\begin{corollary}                                       \label{cor:zeta-thing}
    Consider a randomly chosen CSP with generalized $k$-ary predicate~$P$, as in \pref{lem:random-big-csp}.  Suppose we restrict attention to global assignments $f : \calX \to D$ with the property that
    \begin{equation}    \label{eqn:zeta}
        \mu \coloneqq \absryan{\avg_{x \in \calX}[f(x)]} \leq \zeta,
    \end{equation}
    with $\mu^2 k < \tfrac{1}{2}$.
    Then with high probability, the optimum objective value of the CSP is at most~$\E[P] + \eta_1 + \eta_2$, where $\eta_2 = O(\sqrt{k}\zeta)$.
\end{corollary}
\begin{proof}
    The objective value of achieved by a global assignment~$f$ is $\E_{\bU \sim \Psi_f}[P(\bU)]$.  By \pref{lem:random-big-csp}, this is within $\pm \eta_1$ of $\E_{\bU \sim \phi_f^{\otimes k}}[P(\bU)]$. In turn, multilinearity of~$P$ implies that this equals $P(\mu, \mu, \dots, \mu)$.
    Now,
    \[
    P(\mu,\mu,\dots,\mu) = \E[P] + \iprod{P_{>0},\mu^{\le k}} \le \E[P] +\|P_{>0}\|_2\cdot \|\mu^{\le k}\|,
    \]
    where $P_{>0}$ is $P$ minus its expectation, and $\mu^{\le k}$ is the vector of monomials of degrees $[1,k]$ evaluated at $\mu$, and we have used the Cauchy-Schwarz inequality.
As $P$ is a $k$-ary generalized Boolean predicate, Parseval's theorem implies that $\|P\|_2 \le 1$. At the same time we have that $\|\mu^{\le k}\|_2 \le \sqrt{2k}\mu$, since $\|\mu^{\le k}\|_2^2 \le \sum_{i=1}^k k^i\mu^{2i}$, and we have assumed that $\mu^2 k < \tfrac{1}{2}$. The conclusion follows.
\end{proof}

\begin{corollary}                                       \label{cor:zeta-thing2}
    In the setting of \pref{cor:zeta-thing}, suppose we allow global assignments $f : \calX \to [-1,+1]$ satisfying \pref{eqn:zeta}.  Then with high probability, the optimum objective value of the CSP is at most~$\E[P] + \eta$, where
    \begin{equation}    \label{eqn:eta-defined}
	\eta = O\parens*{\sqrt{\frac{\log (k \Delta)}{\Delta}} + \sqrt{k} \zeta + \frac{k^2}{n}}.
    \end{equation}
\end{corollary}
\begin{proof}
    This follows by taking $D$ to be all integer multiples of $1/\eps$ for $\eps = \Theta\parens*{\frac{1}{\sqrt{k}} \cdot \sqrt{\frac{\log (k \Delta)}{\Delta}}}$.
    Then $|D| = O(1/\eps)$, and we accumulate the error $\eta_1 = O\parens*{\sqrt{\frac{\log (k \Delta)}{\Delta}}+ \frac{k^2}{n}}$, the error $\eta_2 = O(\sqrt{k}\zeta)$, and an additional error from the epsilon net (or by considering only integer multiples of $\eps$).
    By the same argument as in \pref{cor:zeta-thing}'s proof, changing each value of an assignment by $\pm \eps$ can only change a predicate's value by $O(\sqrt{k} \eps) = O\parens*{\sqrt{\frac{\log(k\Delta)}{\Delta}}}$.
\end{proof}

We now come to our main construction.
Let $P : \{\pm 1\}^k \to [0,1]$ be a generalized predicate.
Suppose we produce a CSP over variables $x_1, \dots, x_n$ with $m = \Delta n$ randomly chosen $P$-constraints.
We would like to enforce the constraint that for all $i$, $\absryan*{\E x_i} \le \zeta$ for a small $\zeta$.
Instead, we replace each variable $x_i$ with a ``block'' of $R$ copies, $x_i^{(1)}, \dots, x_i^{(R)}$, and we replace each constraint $P(x_{i_1}, \dots, x_{i_k})$ with all $R^k$ possible constraints of the form $P(x_{i_1}^{(j_1)}, \dots, x_{i_k}^{(j_k)})$.  Finally, we impose a ``close-to-balanced'' constraint:
\begin{equation}    \label{eqn:zeta-final}
    \absryan*{\avg_{i \in [n], j \in [R]} x_{i}^{(j)}}\leq \zeta,
\end{equation}
and call the resulting CSP~$\bC$.  Then:
\begin{theorem}                                     \label{thm:soundness1}
    In the above scenario, with high probability the optimum objective value of~$\bC$ is at most~$\E[P] + \eta$, where $\eta$ is as in \pref{eqn:eta-defined}.
\end{theorem}
\begin{proof}
    Given any assignment to the $x_i^{(j)}$'s, let $\mu_i = \avg_{j \in [R]} x_i^{(j)}$.   Then the global cardinality constraint \pref{eqn:zeta-final} translates to the condition in \pref{eqn:zeta}, and the contribution to the objective function from all the constraints $P(x_{i_1}^{(j_1)}, \dots, x_{i_k}^{(j_k)})$ produced from the ``original constraint'' $P(x_{i_1}, \dots, x_{i_k})$ is precisely $P(\mu_{i_1}, \dots, \mu_{i_k})$.  The result now follows from \pref{cor:zeta-thing2}.
\end{proof}

\subsection{Completeness results}   \label{sec:completeness}

Fix an integer $c \geq 0$ and let $h = \binom{c}{2}$. Let $X_{+2} \subset \{\pm 1\}^4$ be the set of strings of sum~$+2$ and let $X_{-2} \subset \{\pm 1\}^4$ be the set of strings of sum~$-2$. Define a probability distribution~$\nu_c$ on $c^2 \times 4$ matrices in $\{\pm 1\}^{c^2 \times 4}$ as follows. To draw matrix $\bU \sim \nu$:
\begin{enumerate}
    \item with probability $1/2$ choose $h+c$ out of $[c^2]$ rows at random and fill each row with a random string from $X_{+2}$, filling the remaining rows with random strings from~$X_{-2}$;\label{case:a}
    \item with probability $1/2$, do it the other way around, filling $h+c$ random rows with strings from~$X_{-2}$ and the remaining $h$ rows with strings from~$X_{+2}$. \label{case:b}
\end{enumerate}
Note that $c^2 = h+c+h$.
\begin{fact}                                        \label{fact:nu-is-pairwise}
    The distribution $\nu_c$ is pairwise uniform on its $4c^2$ bits.
\end{fact}
\begin{proof}
    By symmetry, each individual bit is uniformly distributed. Thus we only need to show that for any two fixed entries of the matrix, their expected product under~$\nu_c$ is~$0$.  If these two entries are in the same row, this follows from the fact that the uniform distributions over $X_{+2}$ and $X_{-2}$ are pairwise uniform.  Thus it remains to check the case when the two entries are in different rows.  In this case, the expected product of the two bits is
    \begin{align*}
	{} & \Pr(\text{case}~\ref{case:a})\bigg(\Pr(\text{both}\sim X_{+2} | \text{ c.}~\ref{case:a})\left(+\tfrac{1}{2}\right)^2 + \Pr(X_{+2} \text{ and } X_{-2}| \text{ c.}~\ref{case:a}) \left(+\tfrac{1}{2}\right) \left(-\tfrac{1}{2}\right)+ \Pr(\text{both}\sim X_{-2}| \text{ c.}~\ref{case:a})\left(-\tfrac{1}{2}\right)^2\bigg)\\
	{} &+ \Pr(\text{case}~\ref{case:b})\bigg(\Pr(\text{both}\sim X_{-2} | \text{ c.}~\ref{case:b})\left(-\tfrac{1}{2}\right)^2 + 2\Pr(X_{+2}\text{ and } X_{-2}| \text{ c.}~\ref{case:b}) \left(+\tfrac{1}{2}\right) \left(-\tfrac{1}{2}\right)+ \Pr(\text{both}\sim X_{+2}| \text{ c.}~\ref{case:b})\left(+\tfrac{1}{2}\right)^2\bigg)\\
	\\
        {} &= \frac12 \parens*{\frac{h+c}{c^2}\cdot\frac{h+c-1}{c^2-1}\cdot(+\tfrac12)^2 + 2\frac{h+c}{c^2}\cdot\frac{h}{c^2-1} \cdot (+\tfrac12)(-\tfrac12) + \frac{h}{c^2}\cdot\frac{h-1}{c^2-1}\cdot(-\tfrac12)^2} \\
        {} &+ \frac12 \parens*{\frac{h+c}{c^2}\cdot\frac{h+c-1}{c^2-1}\cdot(-\tfrac12)^2 + 2\frac{h+c}{c^2}\cdot\frac{h}{c^2-1} \cdot (-\tfrac12)(+\tfrac12) + \frac{h}{c^2}\cdot\frac{h-1}{c^2-1}\cdot(+\tfrac12)^2} = 0,
    \end{align*}
    as needed.
\end{proof}

The following is now a corollary of \pref{thm:random-graph}:
\begin{theorem}                                     \label{thm:max-bisection-start}
    Suppose we form a random CSP  with $n$ Boolean variables and $m = \Delta n$ ``constraints'' of arity $4c^2$, where to each constraint we associate the distribution~$\nu_c$.  Furthermore, suppose we impose the global constraint saying that the variable assignment must be exactly balanced.  Then with high probability, there is a pseudoexpectation of degree $\Omega(\frac{n}{\Delta^2 \log(1/\Delta)})$ that satisfies the global cardinality constraint and the constraint that all constraint-assignments are distributed according to~$\nu_c$.
\end{theorem}

Take the CSP in the preceding theorem and replace each Boolean variable with~$R$ copies, and each constraint with $R^{4c^2}$ copies, as at the end of \pref{sec:soundness}, producing the random CSP~$\bC'$.  Then it is easy to see that we can extend the SOS pseudoexpectation by treating each copy $x_i^{(j)}$ of $x_i$ as identical.  Thus the extended pseudoexpectation will still satisfy the global cardinality constraint and all $\nu_c$-constraints.

\subsection{Proof of \pref{thm:max-bisection-main}}

We use the random distributional CSP $\bC'$ described in \pref{sec:completeness} to define an instance of max-bisection.
Think of each $4c^2$-ary scope~$S$ as the conjunction of $c^2$ scopes~$S'$ of arity~$4$, each repeated $R^{4c^2}$ times as $S'_{1},\ldots,S'_{R^{4c^2}}$, and associated to the rows of the matrix $\bU$.
Now for each $S'$, add a new Boolean variable $y_{S'}$, and for each $j\in[R^{4c^2}]$ with scope $S_{j}' = (x^{(j)}_{i_1}, \dots, x^{(j)}_{i_4})$ produced, add the four ``$\neq$'' constraints $\{y_{S'} \neq x_{i_1}, y_{S'} \neq x_{i_2}, y_{S'} \neq x_{i_3}, y_{S'} \neq x_{i_4} \}$.
It is important to note that we add only \emph{one} variable $y_{S'}$ for the scope $i_1,i_2,i_3,i_4$, even though this variable participates in $R^{4}$ constraints for each of the $R^{4c^2}$ copies of the $x_{i}$'s.
Finally, impose exact balancedness as a global cardinality constraint (and assume~$c$ is even). This is our final (random) Max-Bisection instance,~$\bM$.

Let us first consider SOS solutions for this instance (completeness).  We retain the SOS pseudoexpectation for~$\bC'$, which gives us a balanced assignment for the $x$-variables, and all $\nu_c$-constraints satisfied.  We extend this pseudoexpectation to the $y_{S'}$ variables as follows.  Let $S$ be $4c^2$-ary scope that contains $S'$ as a $4$-ary scope, and let $S'_1,\ldots,S'_{R^{4c^2}}$ be the $R^{4c^2}$ copies. For the $j$th copy, we define $y^{(j)}_{S'}$ as a degree-$4c^2$ polynomial of the the $x$-variables in~$S_j$. Specifically, for the first~$h$ rows that have sum~$+2$ we define $y_{S'}^{(j)} = -1$, for the first $h$ rows that have sum~$-2$ we define $y^{(j)}_{S'} = +1$, and for the remaining~$c$ rows we define $y_{S'}^{(j)} = \pm 1$ in an alternating fashion.
Then we set $y_{S'} = \sum_{j=1}^{R^{4c^2}}y_{S'}^{(j)}$.
Since each $y^{(j)}_{S'}$ is a weighted sum of $0/1$ indicator functions in the variables of $S$, the $y_{S'}$ are degree-$4c^2$ polynomials in the $x_i$'s.
This loses us a factor of $4c^2$ in the SOS degree,  making it a
\[
    \text{degree-} \Omega\parens*{\frac{n}{\Delta^2 \log(1/\Delta) c^2}} \text{ solution.}
\]
On the other hand, it causes both the $x$-variables and $y$-variables to be exactly balanced with pseudovariance zero, as needed.  Finally, the fraction of the cut constraints that are satisfied will always be precisely
\[
    \frac{\frac34 h + \frac34 h + \frac12 c}{c^2} = \frac34 - \frac{1}{4c}.
\]

As for the true optimal Max-Bisection (soundess), first notice that since we have $Rn$ $x$-variables and $\Delta c^2 n$ $y$-variables, any exactly balanced assignment must be close to balanced on the $x$-side, provided $R \gg \Delta c^2$.  Specifically, a balanced assignment must satisfy \pref{eqn:zeta-final} with $\zeta = \frac{\Delta c^2}{R}$.

Now, using the interpretation of the $4$-ary predicate~$P$ from \pref{eqn:basic-P}, one may deduce from \pref{thm:soundness1} that (with high probability) every true Max-Bisection solution to~$\bM$ has objective value at most
\[
    \frac{11}{16} + O\parens*{\sqrt{\frac{\log (c \Delta)}{\Delta}} + \frac{\Delta c^4}{R} + \frac{c^4}{n}}.
\]
For any constant $c$, this can be made smaller than any $\frac{11}{16} + \eps$ by first taking $\Delta$ a large constant, and then taking~$R$ an even larger constant.  The final number of variables is still $O(n)$, so the final SOS degree is still $\Omega(n)$.  The proof of \pref{thm:max-bisection-main} is complete.

\subsection{Min-Bisection}\label{sec:min-bi}

An almost-identical argument to the proof of \pref{thm:max-bisection-main} yields the following theorem:

\begin{theorem}                                     \label{thm:min-bisection-main}
    For any constant $\eps$ of the form $\frac{1}{4c}$ for $c > 0$ an even integer, there is an easy-to-sample distribution on $n$-vertex Min-Bisection instances such that, with high probability over the choice of instance, the following hold:
    \begin{itemize}
        \item There is a degree-$\Omega(n)$ SOS solution with objective value
         $\frac{1}{4} + \eps$, where the bisection identity and objective value identity are satisfied exactly with pseudovariance zero.
        \item Every genuine solution has objective value at least $\frac{5}{16} - \eps$.
    \end{itemize}
    In particular, degree-$o(n)$ SOS does not deliver a factor-$(\frac{5}{4}-\eps)$-approximation for Min-Bisection.
\end{theorem}

\begin{proof}
    The proof is nearly identical to that of \pref{thm:max-bisection-main}.
    The primary change is that we instead consider the predicate

\begin{equation}    \label{eqn:min-P}
    P(x_1, x_2, x_3, x_4) = \begin{cases}
                                            \frac12 & \text{if $\absryan{\littlesum_i x_i} = 0$,} \\
                                            \frac34 & \text{if $\absryan{\littlesum_i x_i} = 2$,} \\
                                            1 & \text{if $\absryan{\littlesum_i x_i} = 4$.}
                                       \end{cases}
\end{equation}
So that this $P(x_1, x_2, x_3, x_4)$ is the maximum fraction of constraints that can be simultaneously satisfied in the following $5$-variable Max-Uncut instance: $\{y = x_1, y = x_2, y = x_3, y = x_4\}$.
    In this case, $\E[P] = \frac{11}{16}$ as before.

    As above, we use the random distributional CSP $\cC'$ described in \pref{sec:completeness}, but this time we define an instance of Min-Bisection.
    Again we will think of each $4c^2$-ary scope $S$ as the conjunction of $c^2$ scopes $S'$ of arity $4$ (repeated $R^{4c^2}$ times), but now for each scope $S' = (x_{i_1},x_{i_2},x_{i_3},x_{i_4})$ we add a new Boolean variable $y_{S'}$ and the four \emph{equality} constraints $\{y_{S'} = x_{i_1},y_{S'} = x_{i_2},y_{S'} = x_{i_3}, y_{S'} = x_{i_4}\}$. As before we also impose exact balancedness as a global cardinality constraint.

    The completeness analysis is identical to that of \pref{thm:max-bisection-main}, except that we now choose the pseudoexpectation of $y^{(j)}_{S'} = 1$ for the rows that have sum $+2$, $y^{(j)}_{S'} = -1$ for the rows that have sum $-2$, and $y^{(j)}_{S'}$ still alternating $+/- 1$ for the remaining rows.
    We still have the property that the SOS solution is exactly balanced on the $x$ and $y$ variables with pseudovariance zero.
    but now, the fraction of uncut constraints satisfied will be exactly
    \[
	\frac{\frac{3}{4}h + \frac{3}{4} h + \frac{1}{2} c}{c^2} = \frac{3}{4} - \frac{1}{4c},
    \]
which translates to a cut of size $(\frac{1}{4} + \frac{1}{4c})m$ edges,
as desired.

    For the soundness, the analysis is also identical, except that in this case the min-bisection value is at least $1 - (\E[P] + \eta)= \frac{5}{16} - \eta$. This completes the proof.
\end{proof}

%% file: content/conclusions.tex
\section{Conclusions}\label{sec:conclusion}

In this work we have shown that, in the context of random Boolean CSPs, the following strategies do \emph{not} give SOS any additional refutation power: (i)~trying out all possible Hamming weights for the solution; (ii)~trying out all possible (exact) values for the objective function.  We also gave the first known SOS lower bounds for the Min- and Max-Bisection problems.

We end by mentioning some open directions.  There are two technical challenges arising in our work that look approachable.  The first is to extend our results from \pref{sec:localdist} on ``exactifying'' distributions to the case of larger alphabets.  The second is to prove (or disprove) that the ``random*'' and ``purely random'' distributions discussed in \pref{rem:random-star} are~$o(1)$-close (depending on~$m(n)$).

Finally, we suggest investigating further strategies for handling hard constraints in the context of SOS lower bounds.  
Sometimes this is not too difficult, especially when reducing from linear predicates such as 3XOR, where there are perfectly satisfying SOS solutions; we mentioned this in \pref{rem:xor}, and it also arises in, e.g., SOS lower bounds for Densest k-Subgraph \cite{DBLP:conf/soda/BhaskaraCVGZ12, Man15}, where it is essentially automatic that the ``subgraph has size k'' constraint is precisely satisfied.
Other times, it’s of moderate difficulty, perhaps as in this paper’s main \pref{thm:main1} and \pref{thm:main2}.
In still other cases it appears to be very challenging.

One difficult case seems to be in the context of SOS lower bounds for refuting the existence of large cliques in random graphs. 
In \cite{DBLP:conf/focs/BarakHKKMP16} it is shown that in a $G(n, 1/2)$ random graph, with high probability degree-$\Omega(\log n)$ SOS thinks there is a clique of size $\omega := n^{1/2-\epsilon}$.
(Here $\epsilon > 0$ can be any constant.)
 However, it’s merely the case that $\pE[\text{clique size}] \ge \omega$, and it is far from clear how to upgrade the SOS solution so as to actually satisfy the constraint ``clique size = $\omega$'' with pseudovariance zero. 
Besides being an improvement for its own sake, it would be very desirable to
have such an SOS solution for the purposes of further reductions; for example, it would greatly
simplify the recent proofs of SOS lower bounds for approximate Nash welfare in \cite{DBLP:conf/stoc/KothariM18}.
It also seems it might be useful for tackling SOS lower bounds for coloring and stochastic block models.

Finally, we leave as open one more ``hard constraint'' challenge that arises even in the simple context of random 3XOR or 3SAT. 
Suppose one tried to refute random m-constraint 3XOR instances by trying to refute the following statement for all quadruples $(k_{001}, k_{010}, k_{100}, k_{111})$ that sum to $m$:

\medskip
``exactly $k_a$ constraints are satisfied with assignment $a$'', for each $a \in \{001, 010, 100, 111\}$.
\medskip

\noindent As far as we know, constant-degree SOS may succeed with this strategy when $m =  O(n)$.
 It is natural to believe that there is (whp) an $\Omega(n)$-degree SOS pseudodistribution that satisfies all of the above constraints with pseudovariance zero when $k_{001}=  k_{010}=  k_{100}=  k_{111} =  m/4$, but we do not know how to construct one.

\section*{Acknowledgments}
The authors very much thank Sangxia Huang and David Witmer for their contributions to the early stages of this research.  Thanks also to Svante Janson for discussions concerning contiguity of random graph models. We also greatfully acknowledge comments on the manuscript from Johan H{\aa}stad as well as several anonymous reviewers.

%% file: content/kmow-appendix.tex
\section{Revisiting \cite{KMOW17}}  \label{app:kmow}

We first look at the precise expansion property (``Plausibility Assumption'') used in~\cite{KMOW17}.  They introduced the following notions.  A ``subgraph'' $H$ of $G$ refers to an edge-induced subgraph.  It is called a ``$\tau$-subgraph'' if every constraint-vertex in $H$ has degree at least~$\tau$ within~$H$.  We will change this terminology, now calling $H$ a ``$t$-subgraph'' if every constraint-vertex~$a$ in $H$ has degree at least $t_a$ within~$H$.  \cite{KMOW17} also need the notion of a ``$\tau$-subgraph${}^+$'', which is a $\tau$-subgraph together with zero or more isolated variable-vertices; we similarly define ``$t$-subgraph${}^+$''.  Next, for a fixed $\tau$-subgraph${}^+$~$H$, \cite{KMOW17} introduce the following terminology:
\begin{itemize}
    \item Each variable-vertex $v$ is assigned  $2-\deg_H(v)$ ``credits'', where $\deg_H(v)$ denotes the degree of~$v$ within the subgraph~$H$.  (Vertices $v$ with $\deg_H(v)$ are called ``leaves''.)
    \item Each constraint-vertex $a$ is assigned $\tau - \deg_H(a)$ ``debits''.  We modify this terminology so that a constraint-vertex $a$ is assigned $t_a - \deg_H(a)$ ``debits''.
    \item The ``revenue'' of $H$ is defined to be credits minus debits.  The ``cost'' of $H$ is defined to be~$\zeta \cdot |\cons(H)|$, where $\cons(H)$ is the set of constraint-vertices appearing in~$H$.
    \item The ``income'' of $H$ is defined to be revenue minus cost. $H$ is ``plausible'' if its income is nonnegative.
\end{itemize}
Finally, \cite{KMOW17} define the ``Plausibility Assumption'' to be the statement that every $\tau$-subgraph${}$ (equivalently, $\tau$-subgraph${}^+$)~$H$ with $\cons(H) \leq 2 \cdot \CONSMALL$ is plausible.  We will of course change the assumption so that ``$\tau$-subgraph${}^+$'' is replaced with ``$t$-subgraph${}^+$''.

Similar to \cite[Lemma 4.11]{KMOW17}, we can easily check the following:
\begin{definition}  \label{def:income-saturation}
    The income of $H$ can be computed as
    \[
        I = I(H) = T - \zeta c - 2e + 2v,
    \]
    where $T = T(H) = \sum_{a \in \cons(H)} t_a$, $c = c(H) = |\cons(H)|$, $e = e(H)$ is the number of edges in~$H$, and $v = v(H)$ is the number of variable-vertices in~$H$.  Another observation is that if constraint-vertex $a$ appears in $t$-subgraph~$H$, and we add some of $a$'s adjacent edges (that don't already appear) into~$H$, this can only decrease the income of~$H$.  We'll say that $H$ is \emph{constraint-induced} if, for every $a \in \cons(H)$, all of $a$'s adjacent edges are in~$H$.
\end{definition}

Let us now recall the ``closure'' notion from~\cite{KMOW17}.  A subgraph $H$ is called ``small'' if $|\cons(H)| \leq \CONSMALL$.  When $S$ is a set of vertex-variables, $H$ is called ``$S$-closed'' if $H$ is a $\tau$-subgraph and all its leaves are in~$S$.  As usual we modify this to replace ``$\tau$-subgraph'' with ``$t$-subgraph''.  \cite{KMOW17} define the ``closure'' of $S$, $\cl(S)$, to be the union of all small $S$-closed~$H$.  They then define the ``planted distribution $\eta_H$ on small subgraph~$H$'' to be the following probability distribution on vertex-assignments $\bx \in \Omega^n$:  First, for each $a \in \cons(H)$, draw an assignment $\bw_a$ for its neighbors within~$H$ according to $\mu_a$.  Next, for each vertex-variable in~$H$, \emph{condition} on it receiving the same assignment from all its constraint-neighbors within~$H$.  (This is shown to happen with nonzero probability under the Plausibility Assumption.) Finally, assign all other variables uniformly and independently from~$\Omega$.  The key property of $\eta_H$ shown in~\cite{KMOW17} (besides its existence assuming Plausibility) is the following: Whenever $H \supseteq \cl(S)$ is small, the marginal of $\eta_H$ on $S$ is the same as the marginal of $\eta_{\cl(S)}$ on~$S$.

The proof of this key property in~\cite{KMOW17} uses the fact that whenever a constraint $a \in \cons(H)$ has fewer than $\tau$ neighbors in~$H$, a certain expression become~$0$ due to the $(\tau-1)$-wise uniformity of~$\mu_a$; this allows~\cite{KMOW17} to restrict attention to~``$\tau$-subgraphs''.  We do not have the same $(\tau-1)$-wise uniformity for all $\mu_a$ --- rather, we have $(t_a-1)$-wise uniformity for~$\mu_a$ --- but we correspondingly restrict attention to ``$t$-subgraphs''.  Thus the key property goes through with our new definitions.

Next, we describe the pseudoexpectation constructed in~\cite{KMOW17}. For each variable-vertex~$i$ and each $c \in \Omega$, there is an indeterminate $1_c(x_i)$ that is supposed to stand for the $0/1$ indicator that~$i$ is assigned~$c$.  Given a polynomial in these indeterminates, its ``multilinear-degree'' is defined to be the maximum number of variables mentioned in any monomial.  \cite{KMOW17} define their pseudoexpectation~$\pE[\cdot]$ on all polynomials of multilinear-degree at most $\zeta \cdot \CONSMALL$ by imposing that $\pE[p(x)] = \E_{\bx \sim \eta_{\cl(S)}}[p(\bx)]$, where $S$ is the set of variables-vertices mentioned in~$p$.  They show that it satisfies the basic ``Booleanness'' identities $\sum_c 1_c(x_i) = 1$ and $1_c(x_i)^2 = 1_c(x_i)$.  They also show that for every constraint-vertex~$a$, ``the pseudoexpectation's distribution on $a$'s neighbors is always in $\supp(\mu_a)$''; more precisely, that the identity $\sum_{\vec{c} \in \supp(\mu_a)} \prod_{i \sim a} 1_{c_i}(x_i) = 1$ is satisfied.  The proof uses the fact that the set of all edges incident on~$a$ forms a~$\tau$-subgraph.  The analogous statement still holds for our notion of $t$-subgraphs.  (In fact the statement implicitly assumes that $t_a$ is at most the arity of~$a$, but the only way this fails is if $\mu_a$ is the fully-uniform distribution, in which $\supp(\mu_a)$ is full and the desired identity holds vacuously.)

We also jump ahead slightly to the argumentation in~\cite{KMOW17} proving its Theorem~7.2.  There they investigate the constraint-vertices~$a$ for which ``$\pE[\cdot]$'s marginal on $a$'s neighbors is~$\mu_a$'' (as opposed to merely being supported on $\supp(\mu_a)$).  More precisely, these are the constraints~$a$ for which $\pE[ \prod_{i \sim a} 1_{c_i}(x_i)] = \mu_a(\vec{c})$ for all assignments $\vec{c}$ to~$a$'s neighbors.  \cite{KMOW17} shows that for every~$a$ that \emph{doesn't} have this property, there is a corresponding distinct nonempty $\tau$-subgraph~$H$ touching~$a$ (and in fact include all constraint-edges adjacent to~$a$) with $|\cons(H)| \leq 2 \cdot \CONSMALL$ and ``income'' at most $\tau-1$.  Further, distinct such~$a$ yield distinct~$H$. The only way this relies on~$\tau$ is that the star subgraph formed by $a$'s edges has revenue~$\tau$.  In our setup, the revenue is instead~$t_a$.  Thus we obtain the following conclusion: for every constraint-vertex~$b$ where $\pE[\cdot]$'s marginal on $b$'s neighbors \emph{differs} from~$\mu_b$, there is a corresponding nonempty $t$-subgraph~$H_b$ with $|\cons(H_b)| \leq 2 \cdot \CONSMALL$ and income at most~$t_b - 1$; furthermore, the map $b \mapsto H_b$ is an injection.

Finally, we come to the main theorem in~\cite{KMOW17}, its Theorem~6.1, which states that under the Plausibility Assumption, $\pE[\cdot]$ is ``positive semidefinite'' for large degree; specifically, that $\pE[p(x)^2] \geq 0$ provided $\deg(p) \leq \frac13 \zeta \cdot \CONSMALL$.  The proof of this theorem is half a dozen pages; but almost nothing changes when we use the generalized notion of $t$-subgraph in place of $\tau$-subgraph: one only has to check that Lemmas~6.13 and~6.14 continue to hold, and the remainder of the proof goes through.

%% file: content/expansion-appendix.tex
\section{Expansion in random graphs}  \label{app:expansion}

We prove \pref{thm:random-graph} under the assumption that the constraints in~$\calM_2$ all have arity~$2$.
\begin{proof}
    Suppose $H$ is a ``bad'' constraint-induced subgraph, meaning it has $c(H) \leq \CONSMALL$ and it violates \pref{eqn:plaus}; i.e.,
    \begin{equation}    \label{eqn:still-bad}
        v(H) < e(H) - \frac{T(H)}{2} + \zeta c(H).
    \end{equation}

    We begin by showing that a bad~$H$ may be simplified in a way that preserves its badness.  Suppose we have a bad~$H$ induced by constraints $C_1 \subseteq \calM_1$, $C_2 \subseteq  \calM_2$, and $C_t \subseteq \calC$.  Let $V_t = N(C_t)$.  Now suppose $h \in C_1 \cup C_2$ does not touch the variables in~$V$.  Then $h$ does not touch \emph{any} variables in~$H$ beyond its own neighborhood, since $\calM_1$ and $\calM_2$ are nonoverlapping.  Thus deleting~$h$ from~$H$ has the following effects: $e(H)$ and $v(H)$ decrease by the same amount, $c(H)$ decreases by~$1$, and $T(H)$ decreases by either~$1$ or~$2$ (according to whether $h \in C_1$ or $h \in C_2$).  Substituting these observations into \pref{eqn:still-bad} and using $\zeta \leq \frac12$, we see that \pref{eqn:still-bad} remains true.    Thus if there is a bad~$H$ in the factor graph, there is also a ``stripped'' bad~$H$, meaning one that involves no constraints from $\calM_1 \cup \calM_2$ that do not touch its ``$V_t$'' set.

    On the other hand, similar arithmetic shows that, given a stripped~$H$, adding in all constraints from $\calM_1 \cup \calM_2$ that \emph{do} touch~$V_t$ continues to preserve~\eqref{eqn:still-bad}.  Technically, doing this may cause $H$ to cease to be bad if $c(H) > \CONSMALL$; however, we will henceforth consider~$H$ to be bad even if we merely have $|C_t| \leq \CONSMALL$.  Thus in showing the random factor graph is unlikely to contain bad~$H$'s, it suffices to worry about $H$'s with the property that $C_1,C_2$ are precisely determined by $C_t$; i.e., they are the constraints from $\calM_1, \calM_2$ that touch~$V_t = N(C_t)$.

    We now union-bound over all possible bad $H$'s of this form.  Let $|C_t| = c > 0$, let $|C_1| = v_1$, and write $|C_2| = w_1 + w_2$, where $w_1$ is the number of matching constraints touching~$V_t$ at one vertex and~$w_2$ is the number of matching constraints touching~$V_t$ at two vertices.  Thus $|V_t| = v_1 + w_1 + 2w_2$; we will sometimes use that this is at most~$kc$.  For such an~$H$, \pref{eqn:still-bad} is
    \begin{align*}
        v_1 + 2w_1 + 2w_2 &< kc + v_1 + 2w_1 + 2w_2 - \frac{tc}{2} - \frac{v_1}{2} - w_1 -w_2 + \zeta(c+v_1 + w_1 + w_2) \\
        \iff \frac12 v_1 + w_1 + w_2 &< \parens*{k - \frac{t}{2} + \zeta}c + \zeta(v_1 + w_1 + w_2),
    \end{align*}
    which implies
    \begin{equation}    \label{eqn:its-bad}
        \frac12 v_1 + w_1 + w_2 < \parens*{k - \frac{t}{2} + \zeta'}c.
    \end{equation}
    There are $\binom{\Delta n}{c}$ ways to choose $C_t$, and there are  $\binom{U}{v_1} \cdot \binom{(n-U)/2}{w_1, w_2}2^{w_1}$ ways to choose $V_t$.  Let us assume for simplicity that none of $v_1, w_1, w_2$ is~$0$; the the reader may check that the analysis goes through just assuming they are not all~$0$.  We may then use $\binom{a}{b} \leq (\frac{O(1) \cdot a}{b})^b$. Using also $U \leq O(n^{1/2})$, we get that the total number of ways to choose $C_t$ and $V_t$ is at most
    \[
        \parens*{\frac{O(1) \cdot \Delta n}{c}}^c \parens*{\frac{O(1) \cdot n^{1/2}}{v_1}}^{v_1}
        \parens*{\frac{O(1) \cdot n}{w_1}}^{w_1}
        \parens*{\frac{O(1) \cdot n}{w_2}}^{w_2}.
    \]
    Conditioned on $C_t$ and $V_t$, the probability that all $kc$ random edges coming out of $C_t$ fall into the set $V_t$ is at most
    \[
        \parens*{\frac{|V_t|}{n}}^{kc} =
        \parens*{\frac{|V_t|}{n}}^{\parens*{\frac{1}{2}t-\zeta'}c}
        \parens*{\frac{|V_t|}{n}}^{\parens*{k - \frac{1}{2}t + \zeta'}c}
        \leq     \parens*{\frac{O(1) \cdot c}{n}}^{\parens*{\frac{1}{2}t-\zeta'}c}\parens*{\frac{|V_t|}{n}}^{\frac12v_1 + w_1 + w_2}
    \]
    for parameters satisfying \pref{eqn:its-bad}.    Multiplying the above two expressions, we get
    \[
        \E[\#\text{ bad~$H$ with parameters } c, v_1, w_1, w_2] \leq \parens*{O(1) \cdot \Delta \cdot \parens*{\frac{c}{n}}^{\frac{1}{2}t-1-\zeta'}}^c \cdot \frac{|V_t|^{\frac12 v_1 + w_1 + w_2}}{v_1^{v_1} w_1^{w_1} w_2^{w_2}}.
    \]
    Recalling $|V_t| =v_1+w_1 + 2w_2 \leq 2(v_1+w_1+w_2)$, the factor on the right, above, is at most
    \[
        2^{v_1+w_1+w_2} \cdot \frac{(v_1+w_1+w_2)^{v_1 + w_1 + w_2}}{v_1^{v_1} w_1^{w_1} w_2^{w_2}} \leq 6^{v_1 + w_1 + w_2} \leq O(1)^{c},
    \]
    where the first inequality is (the exponentiation of) the log-sum inequality, applied to the sequences $(v_1, w_1, w_2)$ and $(1,1,1)$.  Thus
    \[
    \E[\#\text{ bad~$H$ with parameters } c, v_1, w_1, w_2] \leq \parens*{O(1) \cdot \Delta \cdot \parens*{\frac{c}{n}}^{\frac{1}{2}t-1-\zeta'}}^c,
    \]
    independent of $v_1, w_1, w_2$.  For a given $c$ there are at most $O(c^3)$ choices for $v_1, w_1, w_2$, and this factor can be absorbed into the above expression.  Thus by Markov's inequality, the probability that there exists a bad~$H$ involving exactly~$1 \leq c \leq \CONSMALL$ constraints from $\calC_t$ is at most
    \[
        \parens*{O(1) \cdot \Delta \cdot \parens*{\frac{\CONSMALL}{n}}^{\frac{1}{2}t-1-\zeta'}}^c.
    \]
    Now assuming $\Delta \leq \text{const} \cdot p \cdot \parens*{\frac{n}{\CONSMALL}}^{\frac{1}{2}t-1-\zeta'}$ for a sufficiently small constant, the above is at most $(p/2)^c$ for all~$c$.  Summing this probability over all~$c$ yields a failure probability of at most~$p$, as required.
\end{proof}

%% file: content/imbalanced.tex
\section{Refuting Imbalanced k-XOR Formulae}\label{app:imbal}

Here, we show that imbalanced $k$-XOR formulae are easier to refute.
\begin{proposition}\label{prop:imbal}
	Let $k \ge 3$ be an integer, and let $\Phi$ be a random $k$-XOR instance with $n$ variables $x \in \{\pm 1\}^n$ and $m = \Delta n$ constraints for $m \ge \Omega(\frac{1}{\epsilon^2} n^{(k-1)/2}\log^3 n)$ (where the $O(\cdot)$ notation hides factors depending on $k$).
	Then for any $B$ with $|B| \ge  \epsilon n $, with high probability over $\Phi$ there is no degree-$2(k-1)$ pseudoexpectation for the polynomial system
\[
\{\sum_{i \in [n]} x_i = B\} \cup \{P_j(x) = 1\}_{j \in [m]},
\]
where $P_j(x)$ is the predicate denoting the $j$th constrain of $\Phi$.
\end{proposition}

Recall that in the absence of Hamming weight constraint, a random $k$-XOR instance with $m$ clauses cannot be efficiently refuted by SoS whenever $m = \tilde O(n^{k/2})$ (see \cite{KMOW17}), and our \pref{thm:main1} shows that this remains true whenever there is a Hamming weight constraint of the form $\sum_i x_i = B$ for any $B = O(\sqrt n)$.
By choosing $\epsilon = n^{-(1/4 + \delta}$, a corollary of the above proposition is that this result is not too far from tight.
\begin{corollary}
	If $\Phi = \{P_j\}_{j \in [m]}$ is a $k$-XOR instance with $m  > \tilde{\Omega}(n^{k/2 - 2\delta})$ constraints, then for any $B$ with $|B| \ge n^{1/4 + \delta}$, there is a degree-$2(k-1)$ sum-of-squares refutation of the polynomial system 
	\[
		\{\sum_{i \in [n]} x_i = B\} \cup \{P_j(x) = 1\}_{j \in [m]}.
	\]
\end{corollary}
We conjecture that in fact, one could improve \pref{prop:imbal} so that $m$ depends on $\frac{1}{\eps}$ linearly rather than quadratically, giving that \pref{thm:main1} is tight.

Now we prove the proposition.
\begin{proof}[Proof of \pref{prop:imbal}]
	For convenience, let $P_j(x) = \frac{1}{2}(1 + b_j x^{S_j})$ be the expansion of the predicate $P_j$, with $b_j \in \{\pm 1\}$ and $S_j = \{j_1,\ldots,j_k\} \subset [n]$.
Let $N_i$ denote the number of constraints in which variable $i$ participates in in $\Phi$. 
In the polynomial system above, from the equalities $P_j(x) = 1$ we have that
\begin{align}
	\sum_{i \in [n]} N_i \cdot x_i = \sum_{i \in [n]} \sum_{\substack{j \in [m]\\ x_i \in P_j}} x_i \cdot P_j(x) = \sum_{j \in [m]} (x_{j_1} + \cdots + x_{j_k})\cdot\frac{1}{2}(1 + b_j x_{j_1}\cdots x_{j_k}).
\end{align}
Re-arranging, we have
\begin{align}
	\sum_{i \in [n]} N_i \cdot x_i 
	= \sum_{j \in [m]} b_j\cdot \left(\sum_{i \in [k]} x^{S_j\setminus \{j_i\}}\right)
	= \sum_{i \in [k]} \sum_{j \in [m]} b_j\cdot x^{S_j\setminus \{j_i\}}.
\end{align}
	On the right-hand side, we have the sum of $k$ random (k-1)-XOR instances with $m$ independent clauses each; applying \cite[Corollary 4.2]{DBLP:conf/focs/AllenOW15}, we have that for any degree-$2(k-1)$ or larger pseudoexpectation, the pseudoexpectation of the right-hand side has absolute value at most $k \cdot \gamma m$ for $\gamma = O( \sqrt{\frac{n^{(k-1)/2}\log^3 n}{m}})$.

	Now we show that the left-hand side is close to a multiple of $\sum_i x_i$.
	We have that for every $i \in [n]$, $\E[N_i] = \frac{km}{n} = k\Delta$, and furthermore Bernstein's inequality $\Pr[ |N_i - k\Delta| \ge t ] \le 2\exp(-\frac{t^2}{ k\Delta + \frac{1}{3}t})$.
	Therefore if we let $t = 10\sqrt{k\Delta \log n}$, we have that
\[
	\sum_i N_i x_i = k\Delta \left(\sum_i x_i \right) + \sum_{i} \left(N_i - k\Delta \right) x_i,
\]
	and in the latter quantity, $|\pE[x_i]|\le 1$ and from the union bound with high probability $|N_i - k\Delta| \le 10\sqrt{k\Delta\log n}$ for every $i \in [n]$, so for any pseudoexpectation,
	\[
		\pE\left[\sum_{i\in[n]} N_i x_i\right]
		= k\Delta \pE\left[\sum_{i\in[n]} x_i\right] + \zeta,
	\]
	where $|\zeta| \le 10 n \sqrt{k\Delta\log n}$.

	It follows that for any degree-$2(k-1)$ pseudoexpectation,

	\[
		\pE\left[\sum_{i} x_i\right]
		= \frac{1}{k\Delta}\left( \zeta + \sum_{i \in [k]} \sum_{j\in[m]} b_j \cdot x^{S_j \setminus \{x_{j_i}\}}\right),
		\]
	and the latter quantity is with high probability bounded in absolute value by
	\[
		\left|\pE\left[\sum_{i} x_i\right]\right|
		\le \frac{1}{k\Delta}\left( n\sqrt{100 k \Delta \log n} + k\gamma (\Delta n)\right)
		= n \cdot O\left( \sqrt{\frac{\log^3 n}{k \Delta}} + \sqrt{\frac{n^{(k-1)/2}\log^3 n}{m}}\right).
		\]
		Since we have required $\Delta n = m > \Omega(\frac{1}{\eps^2} n^{(k-1)/2} \log^3 n)$, this proves our claim.
\end{proof}